\documentclass{article}

\usepackage{amsmath,amsfonts,amssymb,amsthm}

\usepackage{float}
\usepackage{color}
\usepackage[all]{xy}

\definecolor{myurlcolor}{rgb}{0,0,0.7}

\usepackage[breaklinks=true]{hyperref}

\hypersetup{colorlinks,
linkcolor=myurlcolor,
citecolor=myurlcolor,
urlcolor=myurlcolor}

\usepackage{graphics}

\usepackage{rotating}

\newcommand{\maps}{\colon}    %correct symbol for colon in f: X -> Y
\newcommand{\R}{{\mathbb R}}  %real numbers
\newcommand{\C}{{\mathbb C}}  %complex numbers
\renewcommand{\H}{{\mathbb H}}  %quaternions
\renewcommand{\O}{{\mathbb O}}  %octonions
\newcommand{\K}{{\mathbb K}}  %a general division algebra
\newcommand{\Z}{{\mathbb Z}}  %integers
\renewcommand{\bar}{\overline} %the conjugate
\newcommand{\h}{\mathfrak{h}} % \h_n(\K) - n x n hermitian matrices with entries in \K.
\newcommand{\SO}{{\rm SO}}    %special orthogonal group
\newcommand{\GL}{{\rm GL}}  %general linear Lie group
\newcommand{\Spin}{{\rm Spin}}    %spin group

\newcommand{\so}{{\mathfrak{so}}}  %special orthogonal Lie algebra
\newcommand{\gl}{{\mathfrak{gl}}}  %general linear Lie algebra
\newcommand{\g}{\mathfrak{g}}  %a generic Lie algebra
\newcommand{\aut}{\mathfrak{aut}} % the automorphism (derivation) Lie algebra
\newcommand{\ext}{\mathfrak{ext}} % the algebra of external symmetries
\renewcommand{\int}{\mathfrak{int}} % the algebra of internal symmetries
\newcommand{\stringIIA}{\mathfrak{string}_{\rm IIA}} % the type IIA string Lie 2-superalgebra
\newcommand{\stringIIB}{\mathfrak{string}_{\rm IIB}}
\newcommand{\mtwobrane}{\mathfrak{m}2\mathfrak{brane}} %the M2-brane Lie 3-algebra

\newcommand{\End}{{\rm End}} %Endomorphisms
\newcommand{\Hom}{{\rm Hom}} %Homomorphisms
\newcommand{\Sym}{{\rm Sym}} %Symmetric algebra/power
\newcommand{\Cl}{{\mathcal{C}\ell}}    %Clifford algebra
\newcommand{\Aut}{\mathrm{Aut}} %automorphism group
\newcommand{\CE}{{\rm CE}} % the Chevalley-Eilenberg complex

\newcommand{\inclusion}{\hookrightarrow}
\newcommand{\tensor}{\otimes} %tensor product
\newcommand{\even}{{\rm even}}

\newcommand{\arxiv}[1]{\href{http://arxiv.org/abs/#1}{arXiv:{#1}}}

\newtheorem{theorem}{Theorem}

\newtheorem{lemma}[theorem]{Lemma}
\newtheorem{proposition}[theorem]{Proposition}

\theoremstyle{definition}
\newtheorem{remark}[theorem]{Remark}
\newtheorem{definition}[theorem]{Definition}
\newtheorem{example}[theorem]{Example}

\newcommand{\be}{\begin{equation}}
\newcommand{\ee}{\end{equation}}
\newcommand{\ba}{\begin{eqnarray}}
\newcommand{\ea}{\end{eqnarray}}
\newcommand{\ban}{\begin{eqnarray*}}
\newcommand{\ean}{\end{eqnarray*}}
\newcommand{\barr}{\begin{array}}
\newcommand{\earr}{\end{array}}

\newcommand{\maximalinvariant}{\star}

\begin{document}

\title{M-Theory from the Superpoint}

\author{
  John Huerta\thanks{CAMGSD, Instituto Superior T\'ecnico, Av.\ Ravisco Pais, 1049-001 Lisboa, Portugal. }\,,
  Urs Schreiber\thanks{Mathematics Institute, Czech Academy of Science, {\v Z}itna 25, 115 67 Praha 1, Czech Republic.}
}

\maketitle

\begin{abstract}
  The ``brane scan'' classifies consistent Green--Schwarz strings and membranes in
  terms of the invariant cocycles on super-Minkowski spacetimes. The ``brane
  bouquet'' generalizes this by consecutively forming the invariant higher central
  extensions induced by these cocycles, which yields the complete fundamental brane content of
  string/M-theory, including the D-branes and the M5-brane, as well as the various
  duality relations between these. This raises the question whether the
  super-Minkowski spacetimes themselves arise as maximal invariant central
  extensions.  Here we prove that they do. Starting from the simplest possible
  super-Minkowski spacetime, the superpoint, which has no Lorentz structure and no
  spinorial structure, we give a systematic process of repeated ``maximal invariant
  central extensions'', and show that it discovers the super-Minkowski spacetimes that
  contain superstrings, culminating in the 10- and 11-dimensional super-Minkowski
  spacetimes of string/M-theory and leading directly to the brane bouquet.

  % \vskip 1em
  % \noindent
  % {\small {\emph{Keywords}: M-theory; rational homotopy theory; cohomology of super Lie algebras.}} \\
  % {\small {\emph{MSC(2010)}: 81T30 (String and superstring theories; other extended
  % objects); 55P62 (Rational homotopy theory); 17B56 (Cohomology of Lie
  % (super)algebras).}}
\end{abstract}

\tableofcontents

\newpage

\section{Introduction}

In his ``vision talk'' at the annual string theory conference in 2014, Greg Moore
highlighted the following open question in string theory \cite[Section 9]{Moore14}:

\begin{quote}
  Perhaps we need to understand the nature of time itself better. [\dots] One natural
  way to approach that question would be to understand in what sense time itself is
  an emergent concept, and one natural way to make sense of such a notion is to
  understand how pseudo-Riemannian geometry can emerge from more fundamental and
  abstract notions such as categories of branes.
\end{quote}

We are going to tell an origin story for spacetime, in which it emerges from the
simplest kind of supermanifold: the superpoint, denoted $\R^{0|1}$. This is the
supermanifold with no bosonic coordinates, and precisely one fermionic
coordinate. From this minimal mathematical space, which has no Lorentz structure and
no spinorial structure, we will give a systematic process to construct
super-Minkowski spacetimes up to dimension 11, complete with their Lorentz structures
and spinorial structures. Indeed this is the same mathematical mechanism that makes,
for instance, the M2-brane and then the M5-brane emerge from 11d spacetime.  It is
directly analogous to the D0-brane condensation by which 11d spacetime emerges out of
the type IIA spacetime of dimension 10.

To make all this precise, first recall that the super $p$-branes of string theory and
M-theory, in their incarnation as `fundamental branes' or `probe branes', are
mathematically embodied in terms of what are called `$\kappa$-symmetric
Green--Schwarz-type functionals'. See Sorokin \cite{Sorokin00} for review and further
pointers.

Not long after Green and Schwarz \cite{GreenSchwarz84} discovered their celebrated
action functional for the superstring, Henneaux and Mezincescu observed
\cite{HenneauxMezincescu85} that the previously somewhat mysterious term in the
Green--Schwarz action, the one which ensures its $\kappa$-symmetry, is in a fact
nothing but the WZW-type functional for super-Minkowski spacetime regarded as a
supergroup. This is mathematically noteworthy, because WZW-type functionals are a
natural outgrowth of super Lie algebra cohomology \cite{AzcarragaIzquierdo,
FSS13}. This suggests that the theory of super $p$-branes is to some crucial extent a
topic purely in super Lie theory, hence amenable to mathematical precision and
classification tools.

Indeed, Azc\'arraga and Townsend \cite{AzcarragaTownsend89} later showed (following
Ach{\'u}carro et al. \cite{AETW87}) that it is the $\Spin(d-1,1)$-invariant super
Lie algebra cohomology of super-Minkowski spacetime which classifies the
Green--Schwarz superstring \cite{GreenSchwarz84}, the Green--Schwarz-type
supermembrane \cite{BergshoeffSezginTownsend87}, as well as all their double
dimensional reductions \cite{DuffHoweInamiStelle87} \cite[Section 2]{FSS16b}, a fact
now known as the ``old brane scan'' \cite{Duff88}.\footnote{ The classification of
these cocycles is also discussed by Movshev et al.\ \cite{MovshevSchwarzXu11} and
Brandt \cite{BrandtII,BrandtIII,BrandtIV}. A unified derivation of the cocycle
conditions is given by Baez and Huerta \cite{BaezHuerta09, BaezHuerta10}. See also
Foot and Joshi \cite{FootJoshi}.  }

For example, for minimal spacetime supersymmetry there is, up to rescaling, a single
non-trivial invariant $(p+2)$-cocycle corresponding to a super $p$-brane in $d$
dimensional spacetime, for just those pairs of $(d,p)$ with $d \leq 11$ that are
marked by an asterisk in the following table.

\begin{table}[H]
\begin{center}

\scalebox{0.9}{
\begin{tabular}{|r||c|c|c|c|c|c|c|c|c|c|}
  \hline
     ${d}\backslash p$ &   & $1$ & $2$ & $3$ & $4$ & $5$ & $6$ & $7$ & $8$ & $9$
	 \\[5pt]
	 \hline \hline
	 $11$ & & &
	  $\star$  &
	  &  &
	 &&& &
	 \\[5pt]
	 \hline
	 $10$ &		
	   & $\star$
		&
		&
		&
		&
       $\star$
	  &
	  &
	  &
	  &
	 \\[5pt]
	 \hline
	 $9$ & & & & & $\star$ & & & & &
	 \\[5pt]
	 \hline
	 $8$  & & & & $\star$ & & & & & &
	 \\[5pt]
	 \hline
	 $7$  & & & $\star$ & & & & & & &
	 \\[5pt]
	 \hline
	 $6$  & & $\star$
		 & & $\star$ & & & & & &
	 \\[5pt]
	 \hline
	 $5$ & &  & $\star$ & & & & & & &
	 \\[5pt]
	 \hline
	 $4$  & & $\star$ & $\star$ &&&&& & &
	 \\[5pt]
	 \hline
	 $3$  & & $\star$ &&&&& & & &
	 \\[5pt]
	 \hline
  \end{tabular}
  }
\end{center}
\caption{The old brane scan.}
\end{table}

Here the entry at $d = 10$ and $p = 1$ corresponds to the Green--Schwarz superstring,
the entry at $d = 10$ and $p = 5$ to the NS5-brane, and the entry at $d = 11$, $p =
2$ to the M2-brane of M-theory fame \cite[Chapter II]{Duff99}. Moving down and to the
left on the table corresponds to double dimensional reduction
\cite{DuffHoweInamiStelle87} \cite[Section 2]{FSS16b}.

This result is striking in its achievement and its failure: On one hand it is
remarkable that the existence of super $p$-brane species may be reduced to a
mathematical classification of super Lie algebra cohomology. But on the other hand,
it is disconcerting that this classification misses so many $p$-brane species that
are thought to exist: The M5-brane in $d = 11$ and all the D-branes in $d = 10$ are
absent from the old brane scan, as are all their double dimensional
reductions.\footnote{ A partial completion of the old brane scan can be achieved by
classifying superconformal structures that may appear in the near horizon geometry of
`solitonic' or `black' $p$-branes \cite{BlencoweDuff88, Duff09}.}

However, it turns out that this problem is not a shortcoming of super Lie theory as such,
but only of the tacit restriction to ordinary super Lie algebras, as opposed to
`higher' super Lie algebras, also called `super Lie $n$-algebras' or `super
$L_\infty$-algebras' \cite{JHThesis, FSS13}.\footnote{Notice that these are Lie
$n$-algebras in the sense of Stasheff \cite{LadaStasheff93, LadaMarkl95, SSS09} as
originally found in string field theory by Zwiebach \cite[Section 4]{Zwiebach92} not
``$n$-Lie algebras'' in the sense of Filippov.  However, the two notions are not
unrelated. At least the Filippov 3-Lie algebras that appear in the Bagger--Lambert model
for coincident solitonic M2-branes may naturally be understood as Stasheff Lie
2-algebras equipped with a metric form \cite[Section 2]{PSa}.}

One way to think of super Lie $n$-algebras is as the answer to the following
question: Since, by a classical textbook fact, 2-cocycles on a super Lie algebra
classify its central extensions in the category of super Lie algebras, what do higher
degree cocycles classify?  The answer (\cite[Prop. 3.5]{FSS13} based on \cite[Theorem
3.1.13]{FiorenzaRogersSchreiber13} and \cite[Theorem 57]{BaezCrans}) is that higher
degree cocycles classify precisely \emph{higher} central extensions, formed in the
homotopy theory of super $L_\infty$-algebras.  But in fact the Chevalley--Eilenberg
algebras for the canonical models of these higher extensions are well known in parts
of the supergravity literature, these are just the ``free differential
algebras''\footnote{Unfortunately, the ``free differential algebras'' of D'Auria and
Fr\'e are not free. In the parlance of modern mathematics, they are differential
graded commutative algebras, where the underlying graded commutative algebra is free,
but the differential is not. We will thus refer to them as ``FDA''s, with quotes.} or
``FDA''s of D'Auria and Fr\'e \cite{DAuriaFre82}.

Hence every entry in the ``old brane scan'', since it corresponds to a cocycle, gives
a super Lie $n$-algebraic extension of super-Minkowski spacetime. Notably the
3-cocycles for the superstring give rise to super Lie 2-algebras and the 4-cocycles
for the supermembrane give rise to super Lie 3-algebras.  These are super-algebraic
analogs of the \emph{string Lie 2-algebra} \cite{BaezCrans} \cite[appendix]{FSS14} which controls the
Green--Schwarz anomaly cancellation of the heterotic string
\cite{TwistedDifferential}, and hence they are called the \emph{superstring Lie
2-algebra} \cite{JHThesis}, to be denoted $\mathfrak{string}$:
$$
\raisebox{9pt}{
\mbox{
\begin{tabular}{|r||c|}
  \hline
     ${d}\backslash p$ &  $1$
	 \\[5pt]
	 \hline \hline
	 $10$ &  $\star$
     \\
	 \hline
  \end{tabular}
}
}
\;\;\;\;\;\;\;\;\;\;\;\; \leftrightarrow
\;\;\;\;\;\;\;\;\;\;\;\;
  \raisebox{20pt}{
  \xymatrix@C=9pt{
    \mathfrak{string}
      \ar[d]
    &
    \mbox{\tiny \raisebox{2pt}{\begin{tabular}{c} extended super Minkowski \\ super Lie 2-algebra \end{tabular}}}
    \\
    \mathbb{R}^{9,1\vert \mathbf{16}}
      &
    \mbox{\raisebox{2pt}{\tiny  \begin{tabular}{c} super Minkowski \\ super Lie algebra\end{tabular} } }
  }
  }
$$
and the \emph{supermembrane Lie 3-algebra}, denoted $\mtwobrane$:
$$
\raisebox{9pt}{
\mbox{
\begin{tabular}{|r||c|}
  \hline
     ${d}\backslash p$ &  $2$
	 \\[5pt]
	 \hline \hline
	 $11$ &  $\star$
     \\
	 \hline
  \end{tabular}
}
}
\;\;\;\;\;\;\;\;\;\;\;\; \leftrightarrow
\;\;\;\;\;\;\;\;\;\;\;\;
  \raisebox{20pt}{
  \xymatrix@C=9pt{
    \mathfrak{m}2\mathfrak{brane}
      \ar[d]
    &
    \mbox{\tiny \raisebox{2pt}{\begin{tabular}{c} extended super Minkowski \\ super Lie 3-algebra \end{tabular}}}
    \\
    \mathbb{R}^{10,1\vert \mathbf{32}}
      &
    \mbox{\raisebox{2pt}{\tiny  \begin{tabular}{c} super Minkowski \\ super Lie algebra\end{tabular} } }
  }
  }
$$
A discussion of these structures as objects in higher Lie theory appears in Huerta's
thesis \cite{JHThesis}. Note that $\mathfrak{string}$ comes in several variants,
denoted $\mathfrak{string}_{\rm IIA}$, $\mathfrak{string}_{\rm IIB}$ and
$\mathfrak{string}_{\rm het}$, corresponding to the type IIA, IIB, and heterotic
variants of string theory. In their dual incarnation as ``FDA''s, the
$\mathfrak{string}$ and $\mtwobrane$ algebras are the extended super-Minkowski
spacetimes considered by Chryssomalakos et al.\ \cite{CAIB00}. We follow their idea,
and call extensions of super-Minkowski spacetime to super Lie $n$-algebras
\emph{extended super-Minkowski spacetimes}.

Now that each entry in the old brane scan is identified with a higher super Lie
algebra in this way, something remarkable happens: \emph{new} cocycles appear on
these extended super-Minkowski spacetimes, cocycles which do not show up on plain
super-Minkowski spacetime itself. (In homotopy theory, this is a familiar phenomenon:
it is the hallmark of the construction of the `Whitehead tower' of a topological
space.)

And indeed, in turns out that the new invariant cocycles thus found do correspond to
the branes that were missing from the old brane scan \cite{FSS13}: On the super Lie
3-algebra $\mtwobrane$ there appears an invariant 7-cocycle, which corresponds to the
M5-brane, on the super Lie 2-algebra $\stringIIA$ there appears a sequence of
$(p+2)$-cocycles for $p \in \{0,2,4,6,8\}$, corresponding to the type IIA
D-branes, and on the superstring Lie 2-algebra $\stringIIB$ there appears a sequence
of $(p+2)$-cocycles for $p \in \{1,3,5,7,9\}$, corresponding to the type IIB
D-branes.  Under the identification of super Lie $n$-algebras with formal duals of
``FDA''s, the algebra behind this statement is in fact an old result: For the
M5-brane and the type IIA D-branes this is due to Chryssomalakos et al.\
\cite{CAIB00}, while for the type IIB D-branes this is due to Sakaguchi \cite[Section
2]{Sakaguchi}.  In fact, the 7-cocycle on the supermembrane Lie 3-algebra that
corresponds to the M5-brane \cite{BLNPST1} was already discovered in the 1982 paper by
D'Auria and Fr\'e \cite[Equations (3.27) and (3.28)]{DAuriaFre82}.

Each of these cocycles gives a super Lie $n$-algebra extension. If we name these
extensions by the super $p$-brane species whose WZW-term is given by the cocycle,
then we obtain the following diagram in the category of super $L_\infty$-algebras:

\vspace{-.6cm}

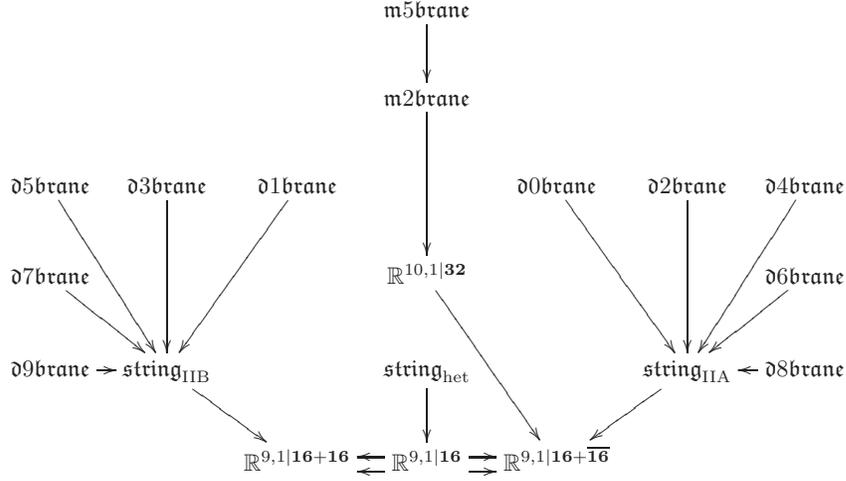
$$
  \vspace{.8cm}
  \scalebox{.9}{
  \xymatrix@C=8pt{
    &
    &&&& { \mathfrak{m}5\mathfrak{brane}}
     \ar[d]
    \\
    &
    &&
     && \mathfrak{m}2\mathfrak{brane}
    \ar[dd]
    &&
    \\
    &
    &
    { \mathfrak{d}5\mathfrak{brane}}
    \ar[ddr]
    &
    {\mathfrak{d}3\mathfrak{brane}}
    \ar[dd]
    &
    {\mathfrak{d}1\mathfrak{brane}}
    \ar[ddl]
    &
    & { \mathfrak{d}0\mathfrak{brane}}
    \ar[ddr]
    &
    { \mathfrak{d}2\mathfrak{brane}}
    \ar[dd]
    &
    { \mathfrak{d}4\mathfrak{brane}}
    \ar[ddl]
    \\
    &
    &
    { \mathfrak{d}7\mathfrak{brane}}
    \ar[dr]
    &
    &
    & \mathbb{R}^{10,1\vert \mathbf{32}}
      \ar[ddr]
    &&&
    { \mathfrak{d}6\mathfrak{brane}}
    \ar[dl]
    \\
    &
    &
    { \mathfrak{d}9\mathfrak{brane}}
    \ar[r]
    &
    \mathfrak{string}_{\mathrm{IIB}}
    \ar[dr]
    &
    & \mathfrak{string}_{\mathrm{het}}
      \ar[d]
    &&
    \mathfrak{string}_{\mathrm{IIA}}
    \ar[dl]
    &
    { \mathfrak{d}8\mathfrak{brane}}
    \ar[l]
    \\
    &
    &
    &
    &
    \mathbb{R}^{9,1 \vert \mathbf{16} + {\mathbf{16}}}
    \ar@{<-}@<-3pt>[r]
    \ar@{<-}@<+3pt>[r]
    & \mathbb{R}^{9,1\vert \mathbf{16}}
    \ar@<-3pt>[r]
    \ar@<+3pt>[r]
    &
    \mathbb{R}^{9,1\vert \mathbf{16} + \overline{\mathbf{16}}}
  }
  }
$$

\vspace{-.4cm}

Hence in the context of higher super Lie algebra, the ``old brane scan'' is completed
to a tree of consecutive higher central extensions emanating out of the
super-Minkowski spacetimes, with one leaf for each brane species in string/M-theory
and with one edge whenever one fundamental brane species may end on another,
with its boundary sourcing a vector- or tensor-multiplet on the worldvolume of the other brane \cite[Section
3]{FSS13}.  This is the fundamental \emph{brane bouquet} \cite[Def.\ 3.9 and Section
4.5]{FSS13}. (The \emph{black} branes and their more general intersection laws are obtained from 
this by passing to equivariant cohomology \cite{ADE}, but this will not concern us here.)

Interestingly, a fair bit of the story of string/M-theory is encoded in this purely
super Lie-$n$-algebraic mathematical structure. This includes in particular the
pertinent dualities: the KK-reduction between M-theory and type IIA theory, the
HW-reduction between M-theory and heterotic string theory, the T-duality between type
IIA and type IIB, the S-duality of type IIB, and the relation between type IIB and
F-theory. All of these are reflected as equivalences of super Lie $n$-algebras
obtained from the brane bouquet \cite{FSS16, FSS16b, ADE}.  The diagram of
super $L_\infty$-algebras that reflects these $L_\infty$-equivalences looks like a
candidate to fill Polchinski's famous schematic picture of M-theory \cite[Figure
1]{Polchinski} \cite[Figure 4]{Witten98} with mathematical life:

\vspace{-1.6cm}

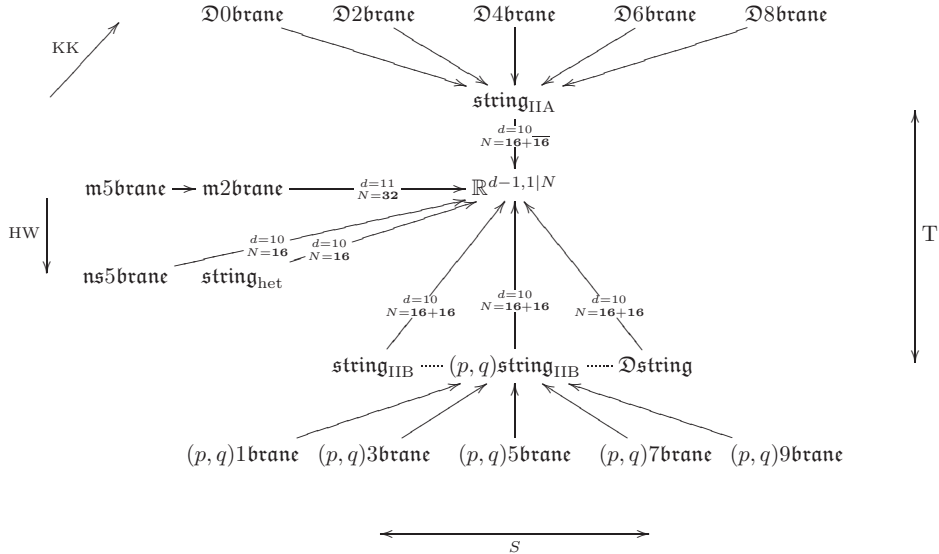
    $$
  \scalebox{.85}{
  \xymatrix@C=2pt{
    && && &&
    \\
    &&
	&& \mathfrak{D}0\mathfrak{brane} \ar[drr]
	& \mathfrak{D}2\mathfrak{brane} \ar[dr]
	& \mathfrak{D}4\mathfrak{brane}\ar[d]
	& \mathfrak{D}6\mathfrak{brane} \ar[dl]
	& \mathfrak{D}8\mathfrak{brane} \ar[dll]
    \\
    & \ar[urr]^{\mathrm{KK}}&
	&
	&
	&& \mathfrak{string}_{\mathrm{IIA}} \ar[d]|(.4){{d=10} \atop {N= \mathbf{16}+ \overline{\mathbf{16}}}}
	&&
	&&
	&&
     \ar@{<->}[ddd]^{\mbox{T}}	
	&&
    \\
    &
    \ar[d]_{\mathrm{HW}}
    && \mathfrak{m}5\mathfrak{brane} \ar[r]
	& \mathfrak{m}2\mathfrak{brane} \ar[rr]|-{{d=11} \atop {N= \mathbf{32}}}
	&& \mathbb{R}^{d-1,1\vert N}
	&&
	&&
	&&
	\\
	&&
    &
    \mathfrak{ns}5\mathfrak{brane} \ar[urrr]|<<<<<<<<<<<<<<<{{d=10}\atop {N = \mathbf{16}}}
	&
    \mathfrak{string}_{\mathrm{het}} \ar[urr]|<<<<<<<{{d=10}\atop {N = \mathbf{16}}}
    \\
    &&&&
	& \mathfrak{string}_{\mathrm{IIB}} \ar[uur]|<<<<<<<<<{{d = 10}\atop {N= \mathbf{16} + \mathbf{16}}}
	\ar@{.}[r]
	& (p,q)\mathfrak{string}_{\mathrm{IIB}} \ar[uu]|<<<<<<<<{{d = 10}\atop {N= \mathbf{16}+ \mathbf{16}}}
	\ar@{.}[r]
	& \mathfrak{Dstring} \ar[uul]|<<<<<<<<<{{d = 10}\atop {N= \mathbf{16} + \mathbf{16}}}
	&&
	&&
	&&
    \\
    &&
	&& (p,q)1\mathfrak{brane} \ar[urr]
	& (p,q)3\mathfrak{brane} \ar[ur]
	& (p,q)5\mathfrak{brane} \ar[u]
	& (p,q)7\mathfrak{brane} \ar[ul]
	& (p,q)9\mathfrak{brane} \ar[ull]
	&
	\\
	&& &&  & \ar@{<->}[rr]_S &  &&
  }
  }
$$

Now note that not all of the super $p$-brane cocycles are of higher degree. One of
them, the cocycle for the D0-brane, is an ordinary 2-cocycle.  Accordingly, the
extension that it classifies is an ordinary super Lie algebra extension. In fact one
finds that the D0-cocycle classifies the central extension of 10-dimensional type IIA
super-Minkowski spacetime to the 11-dimensional spacetime of M-theory.  We can
express these relationships by noting the following diagram of super Lie $n$-algebras
is, in the sense of homotopy theory, a `homotopy pullback':
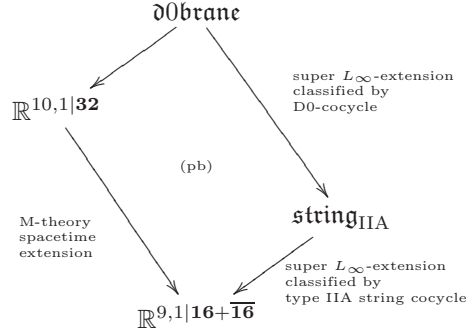
$$
  {
  \xymatrix@C=8pt{
    &
    &
    &
    &
    &
    & { \mathfrak{d}0\mathfrak{brane}}
    \ar@{}[ddd]|{\mbox{\tiny (pb)}}
    \ar[ddr]^{\mbox{\tiny \begin{tabular}{l} super $L_\infty$-extension \\ classified by \\ D0-cocycle\end{tabular}}}
    \ar[dl]
    &
    &
    \\
    &
    &
    &
    &
    & \mathbb{R}^{10,1\vert \mathbf{32}}
      \ar[ddr]_{\mbox{\tiny \mbox{ \begin{tabular}{l} M-theory \\ spacetime \\ extension\end{tabular} }}}
    &&&
    \\
    &
    &
    &
    &
    &
    &&
    \mathfrak{string}_{\mathrm{IIA}}
    \ar[dl]^<<<<{\mbox{\tiny \begin{tabular}{l} super $L_\infty$-extension \\ classified by \\ type IIA string cocycle\end{tabular}}}
    &
    \\
    &
    &
    &
    &
    &
    &
    \mathbb{R}^{9,1\vert \mathbf{16} + \overline{\mathbf{16}}}
  }
  }
$$
This is the precise way to say that the D0-brane cocycle on $\stringIIA$ comes from
pulling back an ordinary 2-cocycle on $\R^{9,1|{\bf 16 + \bar{16}}}$, which in turn
is extended to $\R^{10,1|{\bf 32}}$ by the same 2-cocycle. We may think of this as a
super $L_\infty$-theoretic incarnation of the observation that D0-brane condensation
in type IIA string theory leads to the growth of the 11th dimension of M-theory
\cite[Remark 4.6]{FSS13}, as explained by Polchinski \cite[Section 6]{Polchinski}.

This raises an evident question: Might there be a precise sense in which \emph{all}
dimensions of spacetime originate from the condensation of some kind of 0-branes in
this way?  Is the brane bouquet possibly rooted in the superpoint? Such that the
ordinary super-Minkowski spacetimes, not just extended super-Minkowski spacetimes
such as $\mathfrak{string}$ and $\mtwobrane$, arise from a process of 0-brane
condensation ``from nothing''?

Since the brane bouquet proceeds at each stage by forming \emph{maximal invariant}
extensions, the mathematical version of this question is: Is there a sequence of
\emph{maximal invariant} central extensions that start at the super-point and produce
the super-Minkowski spacetimes in which superstrings and supermembranes exist?

To appreciate the substance of this question, notice that it is clear that every
super-Minkowski spacetime is \emph{some} central extension of a superpoint
\cite[Section 2.1]{CAIB00}: the super-2-cocycle classifying this extension is just
the super-bracket that turns two supercharges into a translation generator. But there
are many central extensions of superpoints that are nothing like super-Minkowski
spacetimes. The question is whether the simple principle of consecutively forming
\emph{maximal invariant} central extensions of super-Lie algebras (as opposed to more
general central extensions) discovers spacetime.

We shall prove that this is the case: this is our main result, Theorem
\ref{thetheorem}.  It says that in the following diagram of super-Minkowski super Lie
algebras, each diagonal morphism is singled out as being the maximal invariant
central extension of the super Lie algebra that it points to:\footnote{ The double
arrows stand for the two different canonical inclusions of $\mathbb{R}^{d-1,1\vert
N}$ into $\mathbb{R}^{d-1,1\vert N + N}$, being the identity on $\mathbb{R}^{d-1,1}$
and sending $N$ identically either to the first or to the second copy in the direct
sum $N + N$.}
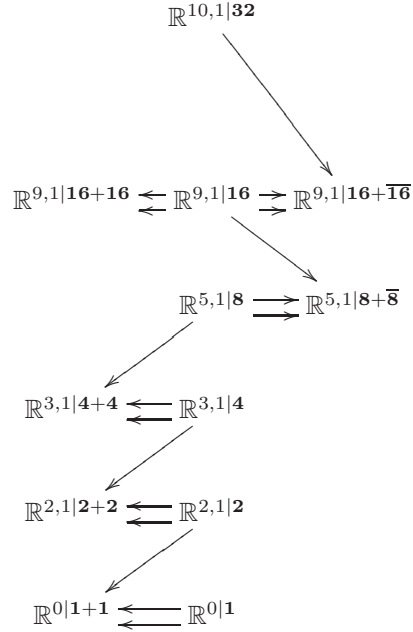
$$
  \vspace{-.3cm}
  \scalebox{1}{
  \xymatrix@C=8pt{
    &
    &
    &
    &
    & \mathbb{R}^{10,1\vert \mathbf{32}}
      \ar[ddr]
    &&&
    \\
    &
    &
    &
    &
    &
    &&
    &
    \\
    &
    &
    &
    &
    \mathbb{R}^{9,1 \vert \mathbf{16} + {\mathbf{16}}}
    \ar@{<-}@<-3pt>[r]
    \ar@{<-}@<+3pt>[r]
    & \mathbb{R}^{9,1\vert \mathbf{16}} \ar[dr]
    \ar@<-3pt>[r]
    \ar@<+3pt>[r]
    &
    \mathbb{R}^{9,1\vert \mathbf{16} + \overline{\mathbf{16}}}
    \\
    &
    &
    &
    &
    &
    \mathbb{R}^{5,1\vert \mathbf{8}}
    \ar[dl]
    &
    \mathbb{R}^{5,1 \vert \mathbf{8} + \overline{\mathbf{8}}}
    \ar@{<-}@<-3pt>[l]
    \ar@{<-}@<+3pt>[l]
    \\
    &
    &
    &
    &
    \mathbb{R}^{3,1\vert \mathbf{4}+ \mathbf{4}}
    \ar@{<-}@<-3pt>[r]
    \ar@{<-}@<+3pt>[r]
    &
    \mathbb{R}^{3,1\vert \mathbf{4}}
    \ar[dl]
    \\
    &
    &
    &
    &
    \mathbb{R}^{2,1 \vert \mathbf{2} + \mathbf{2} }
    \ar@{<-}@<-3pt>[r]
    \ar@{<-}@<+3pt>[r]
    &
    \mathbb{R}^{2,1 \vert \mathbf{2}}
    \ar[dl]
    \\
    &
    &
    &
    &
    \mathbb{R}^{0 \vert \mathbf{1}+ \mathbf{1}}
    \ar@{<-}@<-3pt>[r]
    \ar@{<-}@<+3pt>[r]
    &
    \mathbb{R}^{0\vert \mathbf{1}}
  }
  }
$$

\vspace{.2cm}

Note that we do \emph{not} specify by hand the groups under which these extensions
are to be invariant. Instead these groups are being discovered stagewise, along with
the spacetimes.  Namely we say (Definition \ref{MaximalInvariantCentralExtension})
that an extension $\widehat{\mathfrak{g}} \to \mathfrak{g}$ is \emph{invariant} if it
is invariant with respect to the `simple external automorphisms' inside the
automorphism group of $\mathfrak{g}$ (Definition \ref{internalsymm}). This is a
completely intrinsic concept of invariance.

We show that for $\mathfrak{g}$ a super-Minkowski spacetime, then this intrinsic
group of simple external automorphisms is the spin group, the double cover of the
connected Lorentz group in the corresponding dimension---this is Proposition
\ref{TheAutomorphismLieAlgebras}, read at the Lie group level. This may essentially
be folklore \cite[p. 95]{Evans}, but it seems worthwhile to pinpoint this statement.
It says that as the extension process grows out of the superpoint, not only are
the super-Minkowski spacetimes being discovered as supertranslation supersymmetry
groups, but also their Lorentzian metric structure is being discovered alongside.

\begin{center}
\begin{tabular}{|c|c||c|c|}
  \hline
  \begin{tabular}{c}
    super-Minkowski
    \\
    super Lie algebra
  \end{tabular}
  &
  \begin{tabular}{c}
    simple
    \\
    external automorphisms
  \end{tabular}
  &
  \begin{tabular}{c}
    induced
    \\
    Cartan-geometry
  \end{tabular}
  &
  \begin{tabular}{c}
    torsion
    \\
    freeness
  \end{tabular}
  \\
  \hline
  \hline
  \begin{tabular}{c}
    $\mathbb{R}^{d-1,1\vert N}$
  \end{tabular}
  &
  \begin{tabular}{c}
    $\mathrm{Spin}(d-1,1)$
  \end{tabular}
  &
  \begin{tabular}{c}
    supergravity
  \end{tabular}
  &
  \begin{tabular}{c}
   in $d = 11$:
   \\
   Einstein's equations
  \end{tabular}
  \\
  \hline
\end{tabular}
\end{center}

To highlight this, observe that with every pair $(V,G)$ consisting of a super vector
space $V$ and a subgroup $G \subset \mathrm{GL}(V)$ of its general linear supergroup,
there is associated a type of geometry, namely the corresponding \emph{Cartan
geometry}: A \emph{$(V,G)$-geometry} is a supermanifold with tangent spaces
isomorphic to $V$ and equipped with a reduction of the structure group of its super
frame bundle from $\mathrm{GL}(V)$ to $G$ \cite{Lott90}.

Now for the pairs $(\mathbb{R}^{d-1,1\vert N}, \mathrm{Spin}(d-1,1))$ that emerge out
of the superpoint according to Proposition \ref{TheAutomorphismLieAlgebras} and
Theorem \ref{thetheorem}, this is what encodes a field configuration of
$d$-dimensional $N$-supersymmetric supergravity: Supermanifolds locally modeled on
$\mathbb{R}^{d-1,1\vert N}$ are precisely what underlie the \emph{superspace}
formulation of supergravity, and the reduction of its structure group to the
$\mathrm{Spin}(d-1,1)$-cover of the connected Lorentz group $\SO_0(d-1,1)$ is
equivalently a choice of super-vielbein field, which is a field configuration of
supergravity.

Observe also that the mathematically most natural condition to demand from such a
super-Cartan geometry is that it be `torsion free' \cite{Lott90}. In view of this it
is worthwhile to recall the remarkable theorem of Howe \cite{Howe97}, based on
Candiello and Lechner \cite{CandielloLechner93}: For $d = 11$ the equations of motion
of supergravity are \emph{implied} by the torsion-freeness of the super-vielbein.

In summary, Theorem \ref{thetheorem} shows that the brane bouquet, and with it at
least a fair chunk of the structure associated with the word ``M-theory'', has its
mathematical root in the superpoint, and Proposition \ref{TheAutomorphismLieAlgebras}
adds that as the superspacetimes grow out of the superpoint, they consecutively
discover their relevant Lorentzian metric structure and spinorial structure, and
finally their supergravity equations of motion.

$$
  \vspace{-.3cm}
  \scalebox{1}
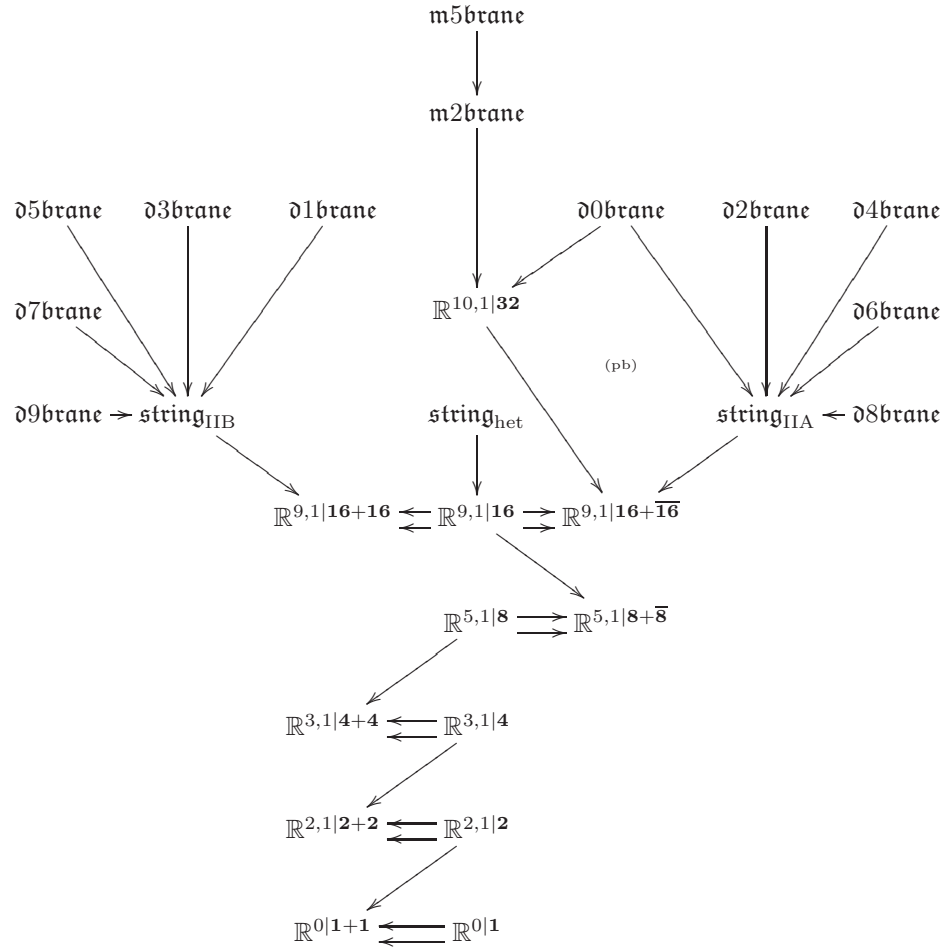
{
  \xymatrix@C=8pt{
    &
    &&&& { \mathfrak{m}5\mathfrak{brane}}
     \ar[d]
    \\
    &
    &&
     && \mathfrak{m}2\mathfrak{brane}
    \ar[dd]
    &&
    \\
    &
    &
    { \mathfrak{d}5\mathfrak{brane}}
    \ar[ddr]
    &
    {\mathfrak{d}3\mathfrak{brane}}
    \ar[dd]
    &
    {\mathfrak{d}1\mathfrak{brane}}
    \ar[ddl]
    &
    & { \mathfrak{d}0\mathfrak{brane}}
    \ar@{}[ddd]|{\mbox{\tiny (pb)}}
    \ar[ddr]
    \ar[dl]
    &
    { \mathfrak{d}2\mathfrak{brane}}
    \ar[dd]
    &
    { \mathfrak{d}4\mathfrak{brane}}
    \ar[ddl]
    \\
    &
    &
    { \mathfrak{d}7\mathfrak{brane}}
    \ar[dr]
    &
    &
    & \mathbb{R}^{10,1\vert \mathbf{32}}
      \ar[ddr]
    &&&
    { \mathfrak{d}6\mathfrak{brane}}
    \ar[dl]
    \\
    &
    &
    { \mathfrak{d}9\mathfrak{brane}}
    \ar[r]
    &
    \mathfrak{string}_{\mathrm{IIB}}
    \ar[dr]
    &
    & \mathfrak{string}_{\mathrm{het}}
      \ar[d]
    &&
    \mathfrak{string}_{\mathrm{IIA}}
    \ar[dl]
    &
    { \mathfrak{d}8\mathfrak{brane}}
    \ar[l]
    \\
    &
    &
    &
    &
    \mathbb{R}^{9,1 \vert \mathbf{16} + {\mathbf{16}}}
    \ar@{<-}@<-3pt>[r]
    \ar@{<-}@<+3pt>[r]
    & \mathbb{R}^{9,1\vert \mathbf{16}} \ar[dr]
    \ar@<-3pt>[r]
    \ar@<+3pt>[r]
    &
    \mathbb{R}^{9,1\vert \mathbf{16} + \overline{\mathbf{16}}}
    \\
    &
    &
    &
    &
    &
    \mathbb{R}^{5,1\vert \mathbf{8}}
    \ar[dl]
    &
    \mathbb{R}^{5,1 \vert \mathbf{8} + \overline{\mathbf{8}}}
    \ar@{<-}@<-3pt>[l]
    \ar@{<-}@<+3pt>[l]
    \\
    &
    &
    &
    &
    \mathbb{R}^{3,1\vert \mathbf{4}+ \mathbf{4}}
    \ar@{<-}@<-3pt>[r]
    \ar@{<-}@<+3pt>[r]
    &
    \mathbb{R}^{3,1\vert \mathbf{4}}
    \ar[dl]
    \\
    &
    &
    &
    &
    \mathbb{R}^{2,1 \vert \mathbf{2} + \mathbf{2} }
    \ar@{<-}@<-3pt>[r]
    \ar@{<-}@<+3pt>[r]
    &
    \mathbb{R}^{2,1 \vert \mathbf{2}}
    \ar[dl]
    \\
    &
    &
    &
    &
    \mathbb{R}^{0 \vert \mathbf{1}+ \mathbf{1}}
    \ar@{<-}@<-3pt>[r]
    \ar@{<-}@<+3pt>[r]
    &
    \mathbb{R}^{0\vert \mathbf{1}}
  }
  }
$$

\vspace{2em}

%%%%%%%%%%%%%%%%%%%%%%%%%%%%%%%%%%%%%%%%%%%%%%%%%%%%%%%%%%%
\section{Automorphisms of super-Minkowski spacetimes}
%%%%%%%%%%%%%%%%%%%%%%%%%%%%%%%%%%%%%%%%%%%%%%%%%%%%%%%%%%%

For our main result, Theorem \ref{thetheorem}, we need to know the automorphisms
(Definition \ref{AutomorphismLieGroup}) of the `super Minkowski super Lie algebras'
$\mathbb{R}^{d-1,1\vert N}$. We give the precise definition of $\R^{d-1,1|N}$ as
Definition \ref{def:SuperMinkowski}, but for the reader's convenience, we quickly
recall the idea. The super-Minkowski super Lie algebra $\R^{d-1,1|N}$ is the version
of Minkowski spacetime, $\R^{d-1,1}$, used when discussing supersymmetry. Unlike
Minkowski spacetime, which is merely a vector space, super-Minkowski is a super Lie
algebra: it has an underlying vector space that is $\Z_2$-graded, with an even and
odd part:
\[ \R^{d-1,1|N}_{\rm even} = \R^{d-1,1}, \quad \R^{d-1,1|N}_{\rm odd} = N. \]
Here, $\R^{d-1,1}$ is ordinary Minkowski spacetime, while $N$ is a spinor
representation of $\Spin(d-1,1)$. The Lie group $\Spin(d-1,1)$ is the double cover of the
connected Lorentz group, $\SO_0(d-1,1)$, so it also acts on $\R^{d-1,1}$. The Lie
bracket on this super Lie algebra is nonzero only on $N$, and consists of a pairing
turning spinors into vectors:
\[ [-,-] \maps N \tensor N \to \R^{d-1,1} \]
which is required to be an equivariant map between representations of $\Spin(d-1,1)$.

Our key idea is that we can extract the Lorentz symmetries of $\R^{d-1,1\vert N}$
merely from its structure as a super Lie algebra, by looking at a particular piece of
the automorphisms we call the `simple external automorphisms'. This result may be
folklore (see Evans \cite[p. 95]{Evans}), but since we did not find a full account in
the literature, we provide a proof here. After some simple lemmas, the result is
Proposition \ref{TheAutomorphismLieAlgebras}. To begin, we define the `simple
external automorphisms' of a super Lie algebra.

\begin{definition}[external and internal automorphisms, admissible algebras]
  \label{internalsymm}
  Let $\mathfrak{g}$ be a super Lie algebra (Def. \ref{SuperLieAlg}), and let
  $\aut(\g)$ be the ordinary Lie algebra of infinitesimal automorphisms of $\g$ which
  preserve the $\Z_2$-grading (Prop. \ref{automorphismLieAlg}).  We define the Lie
  algebra $\int(\g)$ of \emph{internal automorphisms} of $\g$ as the Lie
  subalgebra of $\aut(\g)$ which acts trivially on the even part $\g_{\rm even}$. In
  other words, it is the Lie subalgebra of even derivations of $\g$ which vanish on
  $\mathfrak{g}_{\mathrm{even}}$. This is clearly an ideal, so that the quotient
  $$
    \mathfrak{ext}(\mathfrak{g}) := \mathrm{aut}(\mathfrak{g})/\mathfrak{int}(\mathfrak{g})
  $$
  of all automorphisms by internal ones is again a Lie algebra, the Lie algebra of
  \emph{external automorphisms} of $\mathfrak{g}$. We thus have a short exact sequence:
  $$
    0 \to
    \mathfrak{int}(\mathfrak{g}) \to \mathfrak{aut}(\mathfrak{g})
     \to
    \mathfrak{ext}(\mathfrak{g})
     \to 0
     \,.
  $$
  We will say that $\g$ is \emph{admissible} if this sequence splits and the external
  automorphism algebra $\ext(\mathfrak{g})$ is reductive. For an admissible algebra $\g$,
  we can thus view $\ext(\g)$ as a subalgebra of $\aut(\g)$. Moreover, because we
  demand $\ext(\g)$ be reductive, $\ext(\g)$ decomposes as a direct sum of its center
  and its maximal semisimple Lie subalgebra. We thus define the \emph{simple external
  automorphisms}
  $$
    \mathfrak{ext}_{\mathrm{simp}}(\mathfrak{g})
      \hookrightarrow
    \mathfrak{ext}(\mathfrak{g})
      \hookrightarrow
    \mathfrak{aut}(\mathfrak{g})
  $$
  to be the semisimple part of $\ext(\g)$. 
\end{definition}
\begin{example}
  \label{symmetryR}
  The internal automorphisms (Definition \ref{internalsymm}) of the super-Minkowski
  super Lie algebra $\mathbb{R}^{d-1,1\vert N}$ are the `R-symmetries' from the
  physics literature \cite[p. 56]{FreedLectures}.
\end{example}

Because the super-Minkowski super Lie algebra $\R^{d-1,1|N}$ is built from
$\Spin(d-1,1)$ representations and $\Spin(d-1,1)$-equivariant maps, $\Spin(d-1,1)$
acts on this super Lie algebra by automorphism. It thus acts on the full automorphism
group $\Aut(\R^{d-1,1|N})$ by conjugation, and on the Lie algebra
$\aut(\R^{d-1,1|N})$ by the adjoint action. These facts are key for our first lemma.

\begin{lemma}
  \label{AutomorphismsAsAGraph}
  Consider a super-Minkowski super Lie algebra $\mathbb{R}^{d-1,1\vert N}$
  (Definition \ref{def:SuperMinkowski}) in any dimension $d \geq 3$ and for any real spinor
  representation $N$ of $\Spin(d-1,1)$. Then the automorphism Lie algebra
  $\mathfrak{aut}(\mathbb{R}^{d-1,1|N})$ (Proposition \ref{automorphismLieAlg}) is
  the graph of a surjective, $\mathrm{Spin}(d-1,1)$-equivariant Lie algebra
  homomorphism
  $$
    K \maps \mathfrak{g}_s \longrightarrow \mathfrak{g}_v
    \,,
  $$
  where $\g_s \subseteq \gl(N)$ and $\g_v \subseteq \gl(\R^{d-1,1})$ are the
  projections of $\aut(\R^{d-1,1|N}) \subseteq \gl(N) \oplus \gl(\R^{d-1,1})$ onto
  the summands. Here, $K$ is equivariant with respect to the adjoint action of
  $\Spin(d-1,1)$ restricted to $\g_s$ and $\g_v$.

  In particular the kernel of $K$ is the internal automorphism algebra (Definition
  \ref{internalsymm}), also known as the R-symmetry algebra (Example
  \ref{symmetryR}):
  $$
    \mathrm{ker}(K) \simeq \mathrm{int}(\mathbb{R}^{d-1,1\vert N})
    \,.
  $$
\end{lemma}
\begin{proof}
  We will consider the corresponding inclusion at the level of groups
  $$
    \mathrm{Aut}(\mathbb{R}^{d-1,1\vert N})
      \inclusion
    \mathrm{GL}(N) \times \GL(\R^{d-1,1})
  $$
  with projections $G_s \subseteq \GL(N)$ and $G_v \subseteq \GL(\R^{d-1,1})$. The
  result will then follow by differentiation.

  Note that the spinor-to-vector pairing
  $$
    [-,-] \maps N \otimes N \to \R^{d-1,1} 
  $$
  is surjective, because it is a nonzero map of
  $\mathrm{Spin}(d-1,1)$-representations, and $\R^{d-1,1}$ is irreducible for dimension $d \geq 3$. Hence
  for every vector $v \in \R^{d-1,1}$, there is a pair of spinors $\psi,\phi \in N$ such that
  $$
    v = [\psi,\phi].
  $$
  It follows that for any automorphism $(f,g) \in \mathrm{Aut}(\mathbb{R}^{d-1,1\vert
  N}) \subseteq G_s \times G_v$, $g$ is uniquely determined by $f$ because $(f,g)$ is an automorphism:
  $$
    g(v) = [f(\psi), f(\phi)]
    \,.
  $$
  This determines a function $k \maps G_s \to G_v$ sending $f$ to $g$. It is
  surjective by construction of $G_v$, and is a group homomorphism because its graph
  $\mathrm{Aut}(\mathbb{R}^{d-1,1\vert N})$ is a group. Finally, conjugating $(f,g)$
  by an element of $\Spin(d-1,1)$, it is a quick calculation to check that $k$ is
  $\mathrm{Spin}(d-1,1)$-equivariant, using the equivariance of the spinor-to-vector
  pairing $[-,-]$.
\end{proof}

\begin{lemma}
  \label{AutIsgs}
  Let $\g_s$ be as in Lemma \ref{AutomorphismsAsAGraph}. Then $\aut(\R^{d-1,1|N}) \simeq \g_s$ as Lie algebras. 
\end{lemma}
\begin{proof}
  Because $\aut(\R^{d-1,1|N})$ is the graph of the homomorphism $K \maps \g_s \to
  \g_v$ from Lemma \ref{AutomorphismsAsAGraph}, it is isomorphic to the domain of
  this homomorphism, $\g_s$.
\end{proof}
\begin{lemma}
  \label{DecompositionOfGs}
  Let $N$ be a real spinor representation of $\Spin(d-1,1)$ in some dimension
  $d \geq 3$. Then the Lie algebra $\mathfrak{g}_v$ from Lemma \ref{AutomorphismsAsAGraph}
  decomposes as a $\mathrm{Spin}(d-1,1)$-representation into the direct sum of the
  adjoint representation with the trivial representation:
  $$
    \mathfrak{g}_v \;\simeq\; \mathfrak{so}(d-1,1) \oplus \mathbb{R}.
  $$
  Similarly, the Lie algebra $\mathfrak{g}_s$ from Lemma \ref{AutomorphismsAsAGraph}
  decomposes as a direct sum of exterior powers of the vector representation
  $\R^{d-1,1}$:
  $$
    \mathfrak{g}_s \simeq \underset{i}{\oplus} \, \Lambda^{n_i} \R^{d-1,1}
    \,.
  $$
\end{lemma}
\begin{proof}
  First assume that $N$ is a Majorana spinor representation as in Example
  \ref{MajoranaRepresentations}, and consider $\mathfrak{g}_s$.  Since the Majorana
  representation $N$ is a real subrepresentation of a complex Dirac representation
  $\mathbb{C}^{\mathrm{dim}_{\mathbb{R}}(N)}$ there is a canonical
  $\mathbb{R}$-linear inclusion
  $$
    \mathrm{End}_{\mathbb{R}}(N)
      \hookrightarrow
    \mathrm{End}_{\mathbb{C}}(\mathbb{C}^{\mathrm{dim}_{\mathbb{R}}(N)})
    \,.
  $$
  Therefore it is sufficient to note that the space of endomorphisms of the Dirac representation over
  the complex numbers decomposes into a direct sum of exterior powers of the vector representation.
  This is indeed so, thanks to the inclusion:
  $$
    \mathrm{End}_{\mathbb{C}}(\mathbb{C}^{\mathrm{dim}_{\mathbb{R}}(N)  })
    \hookrightarrow
    \mathrm{Cl}(\mathbb{R}^{d-1,1})\otimes \mathbb{C}
    \,.
  $$
  Explicitly, in terms of the Dirac Clifford basis of
  Example \ref{MajoranaRepresentations}, the decomposition is given by the usual component formula:
  $$
    \psi \otimes \overline{\phi}
    \; \mapsto \;
    \overline{\phi}\psi
    +
    \left(\overline{\phi} \Gamma_ a\psi\right) \Gamma^a
    +
    \tfrac{1}{2}
    \left(\overline{\phi} \Gamma_{a b} \psi\right) \Gamma^{a b}
    +
    \tfrac{1}{3!}\left(\overline{\phi} \Gamma_{a_1 a_2 a_3} \psi\right) \Gamma^{a_1 a_2 a_3}
    +
    \cdots
    \,.
  $$

  Now consider $\mathfrak{g}_v$.
  Recall that, by definition, the automorphism group of $\mathbb{R}^{d-1,1\vert N}$ is
  $$
    \mathrm{Aut}(\mathbb{R}^{d-1,1\vert N})
     :=
    \left\{
      (f,g) \in \GL(N) \times \GL(\R^{d-1,1}) 
      \; : \;
        [f(\psi), f(\phi)] = g[\psi,\phi] \mbox{ for } \psi, \phi \in N
    \right\}
  $$
  and its Lie algebra is
  $$
    \mathfrak{aut}(\mathbb{R}^{d-1,1\vert N})
    =
    \left\{
      (X,Y) \in \mathfrak{gl}(N) \oplus \mathfrak{gl}(\R^{d-1,1})
      \; : \;
      [X \psi, \phi] + [\psi, X \phi] = Y[\psi, \phi]
      \mbox{ for } \psi, \phi \in N
    \right\}
    \,.
  $$
  As we noted above, $\mathrm{Aut}(\mathbb{R}^{d-1,1\vert N})$ always contains
  $\mathrm{Spin}(d-1,1))$, acting canonically, since the spinor-to-vector pairing is
  $\mathrm{Spin}(d-1,1)$-equivariant. Another subgroup of automorphisms that exists
  generally is a copy of the multiplicative group of real numbers $\mathbb{R}^\times$
  where $t \in \mathbb{R}^\times$ acts on spinors $\psi$ as rescaling by $t$ and on
  vectors $v$ as rescaling by $t^2$:
  $$
      \psi \mapsto t \psi, \quad v \mapsto t^2 v. 
  $$
  The Lie algebra of this scaling action is the scaling derivations of Example \ref{scaling}.
  Hence for all $d$ and $N$ we have the obvious Lie algebra inclusion
  $$
    \mathfrak{so}(d-1,1) \oplus \mathbb{R}
      \hookrightarrow
    \mathfrak{aut}(\mathbb{R}^{d-1,1\vert N})
    \,.
  $$
  This shows that there is an inclusion
  $$
    \mathfrak{so}(d-1,1) \oplus \mathbb{R} \hookrightarrow \mathfrak{g}_v \hookrightarrow \mathfrak{gl}(\mathbb{R}^{d})
    \,.
  $$
  Hence it now only remains to see that there is no further summand in $\mathfrak{g}_v$.
  But we know that there is at most one further summand in $\mathfrak{gl}(\R^{d-1,1})$,
  since this decomposes in the form
  $$
    \mathfrak{gl}(\R^{d-1,1})
       \simeq
    \mathfrak{so}(d-1,1) \oplus \mathbb{R} \oplus
      \Sym^2_0(\R^{d-1,1}), 
  $$
  where $\Sym^2_0(\R^{d-1,1})$ denotes the space of traceless, symmetric $d \times d$
  matrices. It follows that the only further summand that could appear in
  $\mathfrak{g}_v$ is $\Sym^2_0(\R^{d-1,1})$.  But by Lemma
  \ref{AutomorphismsAsAGraph}, the homomorphism $K \maps \mathfrak{g}_s \to
  \mathfrak{g}_v$ is surjective, so its image $\g_v$ must be a subset of the exterior
  powers appearing in $\g_s$.  Since the symmetric traceless matrices and the
  exterior powers $\Lambda^\bullet \mathbb{R}^d$ are distinct irreducible
  representations of $\Spin(d-1,1)$, we conclude $\Sym^2_0(\R^{d-1,1})$ is not a summand of
  $\g_v$.

  This concludes the proof for the case that $N$ is a Majorana representation. The
  argument for $N$ symplectic Majorana (Ex. \ref{MajoranaRepresentations}) is
  similar. Finally, a general real spin representation is a direct multiple of $N$ or
  a sum of multiples of the two Weyl representations $N \simeq N_- \oplus N_+$. We
  generalize to these cases in turn.

  First, we consider $nN$, a direct multiple of $N$, for $n$ some nonnegative
  integer. Since $\End_\R(nN) \simeq n^2 \End_\R(N)$, the left hand side is indeed a
  sum of exterior powers.

  Next, if $N$ decomposes as $N_- \oplus N_+$, a general spin representation is of
  the form $n_- N_- \oplus n_+ N_+$, for $n_-$ and $n_+$ nonnegative integers. We wish to show that
  \[ \End_\R(n_- N_- \oplus n_+ N_+) \simeq n_-^2 \End_\R(N_-) \, \oplus \, n_- n_+ \Hom_\R(N_-, N_+) \, \oplus \, n_+ n_- \Hom_\R(N_+, N_-) \, \oplus \, n_+^2 \End_\R(N_+)  \]
  is a sum of exterior powers. Yet we have already shown that
  \[ \End_\R(N_- \oplus N_+) \simeq \End_\R(N_-) \, \oplus \, \Hom_\R(N_-, N_+) \,
  \oplus \, \Hom_\R(N_+, N_-) \, \oplus \, \End_\R(N_+) \]
  is a sum of exterior powers. Thus, every summand on the right hand side is a sum of
  exterior powers, and it follows that $\End_\R(n_-N_- \oplus n_+ N_+)$ is also.

 \end{proof}

\begin{proposition}
  \label{TheAutomorphismLieAlgebras}
  For any dimension $d \geq 3$ and real spinor representation $N$ of $\Spin(d-1,1)$, the
  super-Minkowski super Lie algebra $\mathbb{R}^{d-1,1\vert N}$ (Definition
  \ref{def:SuperMinkowski}) is admissible (Definition \ref{internalsymm}). Moreover,
  the Lie algebra of external automorphisms (Definition \ref{internalsymm}) of
  $\R^{d-1,1|N}$ is the direct sum:
  $$
    \mathfrak{ext}(\mathbb{R}^{d-1,1\vert N})
     \;\simeq\;
     \mathfrak{so}(d-1,1) \oplus \mathbb{R},
  $$
  where $\so(d-1,1)$ acts in the canonical way on $\R^{d-1,1|N}$ (Definition
  \ref{def:SuperMinkowski}) and $\R$ acts by the scaling action from Example
  \ref{scaling}.
\end{proposition}
\begin{proof}
  The admissibility of $\R^{d-1,1|N}$ will follow when we determine
  $\ext(\R^{d-1,1|N})$ has the form claimed, since this form is reductive, and the
  action of $\ext(\R^{d-1,1|N})$ on $\R^{d-1,1|N}$ described in the proposition gives
  the splitting. So we prove this form is correct. By Lemma \ref{AutIsgs}, we have
  $\mathfrak{aut}(\mathbb{R}^{d-1,1\vert N}) \simeq \mathfrak{g}_s$ and by Lemma
  \ref{DecompositionOfGs} we have a decomposition as $\Spin(d-1,1)$-representations
  $$
    \mathfrak{aut}(\mathbb{R}^{d-1,1\vert N})
      \;\simeq\;
    \left(
    \mathfrak{so}(d-1,1)
      \;\oplus\;
    \mathbb{R}
    \right)
      \;\oplus\;
    \underset{= \mathfrak{int}(\mathbb{R}^{d-1,1\vert N})}{\underbrace{\mathrm{ker}(K)}}
    \,,
  $$
  where the last summand is the algebra of internal automorphisms (Definition
  \ref{internalsymm}), hence the R-symmetries (Example \ref{symmetryR}). Therefore
  the claim follows by Definition \ref{internalsymm}.
\end{proof}

%\begin{proposition}
%  For $d \in \{3,4,6,10,11\}$ and for the irreducible (symplectic/Weyl) Majorana $\mathrm{Spin}(d-1,1)$-representation $N$
%  (example \ref{MajoranaRepresentations}) the R-symmetry Lie algebra is as follows:
%  \begin{center}
%  \begin{tabular}{|cc|l|}
%    \hline
%    $d$ & $N$ & $\mathfrak{aut}(\mathbb{R}^{d-1,1\vert N})$
%    \\
%    \hline
%    \hline
%    $3$ & $\mathbf{2}$ & $\mathfrak{so}(\phantom{1}2,1) \oplus \mathbb{R}$
%    \\
%    \hline
%    $4$ & $\mathbf{4}$ & $\mathfrak{so}(\phantom{1}3,1) \oplus  \mathbb{R} \oplus \;\mathfrak{u}(1)$
%    \\
%    \hline
%    $6$ & $\mathbf{8}$ &  $\mathfrak{so}(\phantom{1}5,1) \oplus \mathbb{R} \oplus  \mathfrak{su}(2)$
%    \\
%    \hline
%    $10$ & $\mathbf{16}$  & $\mathfrak{so}(\phantom{1}9,1)  \oplus \mathbb{R}$
%    \\
%    \hline
%    $11$ & $\mathbf{32}$ & $\mathfrak{so}(10,1) \oplus \mathbb{R}$
%    \\
%    \hline
%  \end{tabular}
%  \end{center}
%\end{proposition}

\section{The maximal invariant central extensions of the superpoint}
\label{TheConsecutiveExteension}

With the results from the previous section in hand, we have a way of talking about
the $\so(d-1,1)$ symmetries of a super-Minkowski super Lie algebra $\R^{d-1,1|N}$
\emph{purely in terms of its Lie bracket}: it is the algebra of simple external
automorphisms of $\R^{d-1,1|N}$, by Proposition
\ref{TheAutomorphismLieAlgebras}. This allows us to begin with a super Lie algebra
that lacks any apparent relation to spacetime, and discover spacetime symmetries via
the automorphisms. We make repeated use of this in our construction of
super-Minkowski spacetimes by central extension of the superpoint, $\R^{0|1}$. This
is our main result, Theorem \ref{thetheorem}.

To be precise, we compute consecutive `maximal invariant central extensions' of the
superpoint. First we state the definition of the extension process:
\begin{definition}[maximal invariant central extensions]
  \label{MaximalInvariantCentralExtension}
  Let $\mathfrak{g}$ be an admissible super Lie algebra (Definitions
  \ref{internalsymm} and \ref{SuperLieAlg}), let $\mathfrak{h} \hookrightarrow
  \mathfrak{aut}(\mathfrak{g})$ be a subalgebra of its automorphism Lie algebra
  (Proposition \ref{automorphismLieAlg}) and let
  $$
    \xymatrix{
      V \ar@{^{(}->}[r] & \widehat{\mathfrak{g}}
      \ar[d]
      \\
      & \mathfrak{g}
    }
  $$
  be a central extension of $\mathfrak{g}$ by a vector space $V$ in even degree. Then we say that
  $\widehat{\mathfrak{g}}$ is
  \begin{enumerate}
  \item an \emph{$\mathfrak{h}$-invariant central extension}
  if the even 2-cocycles that classify the extension, according to Example \ref{2CocyclesAndCentralExts},
  are $\mathfrak{h}$-invariant 2-cocycles as in Definition \ref{InvariantCohomology};
  \item
  an \emph{invariant central extension} if it is
  $\mathfrak{h}$-invariant and $\mathfrak{h} = \mathfrak{ext}_{\mathrm{simp}}(\mathfrak{g})$
  is the semi-simple part of its external automorphism Lie algebra (Definition \ref{internalsymm});
  \item
  a \emph{maximal $\mathfrak{h}$-invariant central extension} if it is an $\mathfrak{h}$-invariant central extension
  such that the $n$-tuple of $\mathfrak{h}$-invariant even 2-cocycles that classifies it (according to Example \ref{2CocyclesAndCentralExts})
  is a linear basis for the even $\mathfrak{h}$-invariant cohomology $H_{\rm even}^2(\mathfrak{g},\mathbb{R})^{ \mathfrak{h} }$ (Definition \ref{InvariantCohomology}).
  \end{enumerate}
  When the central extension $\widehat{\g}$ is both maximal and invariant, we say it is a
  \emph{maximal invariant central extension} and distinguish it with the symbol
  $\maximalinvariant$ on the projection map $\widehat{\g} \to \g$, like this:
  $$ 
    \xymatrix{
      V \ar@{^{(}->}[r] & \widehat{\mathfrak{g}}
      \ar[d]^{\maximalinvariant}
      \\
      & \mathfrak{g}
    }
  $$
\end{definition}

We now begin to climb the tower of maximal invariant central extensions, beginning
with the superpoint. But first we must note that our starting point is admissible.
\begin{lemma}
  The superpoint $\R^{0|N}$ (Def. \ref{Superpoint}) is admissible
  (Def. \ref{internalsymm}) for any natural number $N \in \mathbb{N}$.
\end{lemma}
\begin{proof}
  All automorphisms of $\R^{0|N}$ are internal (Def. \ref{internalsymm}), so the
  external automorphisms are trivial. Hence, they are trivially reductive, and
  trivially a subalgebra of $\aut(\R^{0|N})$, which is what it means to be
  admissible.
\end{proof}

In our next proposition, we see spacetime appear by extending a superpoint.

\begin{proposition}\label{ProofOfTheFirstStep}
  The maximal invariant central extension (Definition
  \ref{MaximalInvariantCentralExtension}) of the superpoint $\mathbb{R}^{0\vert 2}$
  (Def. \ref{Superpoint}) is the 3-dimensional super Minkowski super Lie algebra
  $\mathbb{R}^{2,1\vert \mathbf{2} }$ as in Definition \ref{def:SuperMinkowski}:
  $$
    \xymatrix{
      \mathbb{R}^3 \ar@{^{(}->}[r] &  \mathbb{R}^{2,1\vert \mathbf{2}} \ar[d]^{\maximalinvariant}
      \\
      & \mathbb{R}^{0\vert 2}
    }
  $$
  with $N = \mathbf{2}$ the unique irreducible real spinor representation of $\mathrm{Spin}(2,1)$ from
  Proposition \ref{RealSpinorsFromTheNormedDivisionAlgbras}.
\end{proposition}
\begin{proof}
  Since $\R^{0|2}$ is concentrated in odd degree, the external automorphisms are
  trivial: $\mathfrak{ext}(\mathbb{R}^{0\vert 2}) = 0$. Thus, every central extension
  is invariant (Def.  \ref{MaximalInvariantCentralExtension}).

  According to Example \ref{2CocyclesAndCentralExts}, the maximal central extension
  is the one induced by the all of the even super Lie algebra 2-cocycles on
  $\mathbb{R}^{0\vert 2}$.  Since $\mathbb{R}^{0\vert 2}$ is concentrated in odd
  degree and has trivial Lie bracket, an even 2-cocycle in this case is given by a
  symmetric bilinear form on $\mathbb{R}^2$. There is a 3-dimensional real vector
  space of these.  This shows that the underlying super vector space of the maximal
  central extension is $\mathbb{R}^{3\vert 2}$.  It remains to check that the Lie
  bracket is that of 3d super-Minkowski.

  If we let $\{d \theta_1, d \theta_2\}$ denote the canonical basis of the dual space
  $\mathbb{R}^{{0\vert 2}*}$, then the space of even 2-cocycles is spanned by:
$$
  \raisebox{20pt}{
  \xymatrix@R=7pt@C=5pt{
    d \theta^1 \wedge d \theta^1
      &
    d \theta^1 \wedge d \theta^2
    \\
    & d \theta^2 \wedge d \theta^2
  }
  }
  \, , 
$$
where the wedge product is symmetric between these odd elements. By the formula for
the central extension from Example \ref{2CocyclesAndCentralExts}, this means that the
super Lie bracket is given on the spinors $\psi = \left(
  \raisebox{10pt}{\xymatrix@R=3pt{ \psi_1 \\ \psi_2 }}\right)$ and $\phi = \left(
  \raisebox{10pt}{\xymatrix@R=3pt{ \phi_1 \\ \phi_2 }}\right)$ by
$$
  \begin{aligned}
  \left[
    \psi, \phi
  \right]
& =
  \left(
    \raisebox{10pt}{
    \xymatrix@R=3pt@C=3pt{
      \psi_1 \phi_1 & \tfrac{1}{2}(\psi_1 \phi_2 + \phi_1 \psi_2)
      \\
      \tfrac{1}{2}(\psi_1 \phi_2 + \phi_1 \psi_2) & \psi_2 \phi_2
    }}
   \right)
   \;\;\;=\;\;\;
   \tfrac{1}{2}\left(\psi \phi^\dagger +  \psi \phi^\dagger \right),
   \end{aligned}
$$
Comparing this formula to Proposition \ref{RealSpinorsFromTheNormedDivisionAlgbras},
we see this is indeed the spinor-to-vector pairing for the real representation
$\mathbf{2}$ of $\mathrm{Spin}(2,1)$.
\end{proof}

To deduce the maximal invariant central extensions of $\mathbb{R}^{2,1 \vert 2}$, we
use the representation of spinors via the normed division algebras as a key tool. We give
all the details in Proposition \ref{RealSpinorsFromTheNormedDivisionAlgbras} of our
appendix, but for the reader's convenience, we quickly summarize the idea.

There are four real normed division algebras: the real numbers $\R$, the complex
numbers $\C$, the quaternions $\H$, and the octonions $\O$. They have dimensions 1,
2, 4 and 8, respectively. It is a famous fact that the octonions $\O$ are not
associative, while $\R, \C$ and $\H$ are. For $\K$ a normed division algebra of
dimension $k$, we can construct spinors for spacetime of dimension $k+2$. More
precisely, we can cook up two irreducible, real spinor representations of the spin
group $\Spin(k+1,1)$, the double cover of the connected Lorentz group
$\SO_0(k+1,1)$. Both of these spinor representations are defined on the vector space
$\K^2$, but they differ in the action of $\Spin(k+1,1)$:
\[ N_+ = \K^2, \quad N_- = \K^2 . \]
We can also define Minkowski spacetime itself in terms of $\K$, as the space of $2
\times 2$ hermitian matrices over $\K$:
\[ \R^{k+1, 1} := \left\{ \begin{bmatrix} t + x & \overline{y} \\ y & t -
    x \end{bmatrix} : t, x \in \R, \, y \in \K \right\} \]
where $\overline{y} \in \K$ is denotes the conjugate of $y \in \K$. For the more
details, see Proposition \ref{RealSpinorsFromTheNormedDivisionAlgbras}.

Our next lemma relates the construction of spinors from $\K$ and from the
`Cayley--Dickson double', $\K_{\rm dbl}$ (Def. \ref{CayleyDickson}). Roughly, the
Cayley--Dickson double takes a normed division algebra $\K$ of dimension $k$, and
gives the `next' normed division algebra $\K_{\rm dbl}$ of dimension $2k$:
\[ \R_{\rm dbl} = \C, \quad \C_{\rm dbl} = \H, \quad \H_{\rm dbl} = \O . \]
This process breaks down for $\K = \O$, when $\O_{\rm dbl}$ fails to be a division
algebra.

In any case, the Cayley--Dickson double contains the original algebra as a
subalgebra, $\K \subseteq \K_{\rm dbl}$. This means that the $2 \times 2$ hermitian
matrices over $\K$ are a subset of those over $\K_{\rm dbl}$, and hence there is an
inclusion of spacetimes $\R^{k+1,1} \subseteq \R^{2k + 1, 1}$, and a corresponding
inclusion of spin groups $\Spin(k+1,1) \subseteq \Spin(2k+1,1)$. By restricting along
this inclusion, spinor representations of $\Spin(2k+1,1)$ become representations of
$\Spin(k+1,1)$. The next lemma tells us precisely which representations we obtain in
this way.

\begin{lemma}
  \label{branchingOfCDDoubleSpinReps}
  Let $\mathbb{K} \in \{\mathbb{R}, \mathbb{C}, \mathbb{H}\}$ be an associative
  normed division algebra (Example \ref{RCHO}) of dimension $k$, and let
  $\mathbb{K}_{\mathrm{dbl}} \in \{\C, \H, \O\}$ be its Cayley--Dickson double
  (Definition \ref{CayleyDickson}) of dimension $2k$. Let $N_+$ and $N_-$ be the real
  spinor representations defined in terms of $\K$, and let $N_{{\rm dbl}}$ denote
  either of the real spinor representations defined in terms of $\K_{\rm dbl}$, as in
  Proposition \ref{RealSpinorsFromTheNormedDivisionAlgbras}.  Consider the inclusion
  of spin groups $\Spin(k+1,1) \subseteq \Spin(2k + 1, 1)$ induced by the inclusion
  of normed division algebras $\K \subseteq \K_{\rm dbl}$. Restricting along this
  inclusion, the irreducible real $\mathrm{Spin}( 2k+1,1 )$-representation
  $N_{\mathrm{dbl}}$ branches into the direct sum of the two irreducible real
  $\mathrm{Spin}(k+1,1)$-representations $N_+, N_-$:
  $$
    N_{\mathrm{dbl}} \simeq N_+ \oplus N_-
    \,.
  $$
\end{lemma}
\begin{proof}
  We will prove the result for $N_{{\rm dbl}+}$, as the argument for $N_{{\rm dbl}-}$
  will be similar. By Proposition \ref{RealSpinorsFromTheNormedDivisionAlgbras}, the
  spin representation $N_{\mathrm{dbl}+}$ is defined on the real vector space
  $\mathbb{K}_{\mathrm{dbl}}^2$. By Cayley--Dickson doubling (Definition
  \ref{CayleyDickson}), this is given in terms of $\mathbb{K}$ as the direct sum
  $$
    \mathbb{K}_{\mathrm{dbl}}^2
      \simeq
    \mathbb{K}^2 \oplus \mathbb{K}^2 \ell
    \,.
  $$
  This makes it immediate that the first summand $\mathbb{K}^2$ is $N_+$ as a representation of $\Spin(k+1,1)$. 
  We need to show that the second summand is isomorphic to $N_-$. 

  To that end observe, by the relations in the Cayley--Dickson construction
  (Definition \ref{CayleyDickson}), that for $\psi \in \mathbb{K}^2$ and $A \in
  \h_2(\K)$ a $2 \times 2$ hermitian matrix, we have the following identity:
  $$
    \begin{aligned}
      A (\psi \ell)
      & =
      A (\ell \overline{\psi})
      \\
      & =
      \ell( \overline{A} \, \overline{\psi} )
      \\
      & =
      \ell (\overline{A}_L \overline{\psi})
      \\
      & =
      \ell( \overline{A_R \psi} )
      \\
      & =
      (A_R \psi) \ell
      \,,
    \end{aligned}
  $$
  where $A_L$ and $A_R$ denotes the right and left actions of the matrix $A$,
  respectively (Def. \ref{MatricesToLinearOps}), and we have used
  Prop. \ref{LeftAndRight} to relate left and right actions under conjugation.

  Recall from Prop. \ref{RealSpinorsFromTheNormedDivisionAlgbras} that $\Spin(k+1,1)$
  is the subgroup of the Clifford algebra generated by products of pairs of unit
  vectors of the same sign:
  \[ \Spin(k+1, 1) = \langle AB \in \Cl(k+1,1) \, : \, A, B \in \R^{k+1,1}, \,
  \eta(A,A) = \eta(B,B) = \pm 1 \rangle . \]
  It follows from our above calculation that the action of a generator $AB \in
  \Spin(k+1,1)$ on the summand $\mathbb{K}^2 \ell$ is the composition of right
  actions on $\mathbb{K}^2$:
  $$
    \begin{aligned}
      \tilde A_L B_L (\psi \ell)
      & =
      \tilde A( B (\psi \ell) )
      \\
      & =
      \left( \tilde A_R  B_R (\psi) \right) \ell
    \end{aligned}
  $$
  Therefore we are now reduced to showing that this action of $\Spin(k+1,1)$ on $\K^2$:
  $$
    \psi \mapsto \tilde A_R B_R (\psi) \mbox{ for } \psi \in \K^2
  $$
  is isomorphic to the action of $\Spin(k+1,1)$ on $N_-$, which also has the
  underlying vector space $\K^2$:
  $$
    \psi \mapsto A_L \tilde B_L (\psi) \mbox{ for } \psi \in \K^2
    \,.
  $$
  We claim there is an isomorphism given by
  $$
    F \maps \psi \mapsto J \overline{\psi}
    \,,
  $$
  where $J$ is the matrix:
  $$
    J := \left(
      \begin{array}{cc}
        0 & -1
        \\
        1 & 0
      \end{array}
    \right)
    \,.
  $$
  A quick calculation shows that $J$ satisfies the matrix identity:
  $$
    J \overline{A} = - \tilde A J
  $$
  for any $A \in \h_2(\K)$ a $2 \times 2$ hermitian matrix. We use this to show that
  $F$ is indeed an isomorphism:
  $$
    \begin{aligned}
      F( \tilde A_R B_R (\psi) )
      & =
      J \, \overline{ \tilde A_R B_R (\psi) }
      \\
      & =
      J \, \overline{\tilde A}_L \, \overline{B}_L (\overline{\psi} )
      \\
      & =
      J \, \overline{\tilde A} ( \overline{B} \, \overline{\psi} )
      \\
      & =
      A ( \tilde B J \overline{\psi} )
      \\
      & =
      A_L \tilde B_L (F(\psi))
    \end{aligned}
  $$
\end{proof}

Next, for our spinor representation $N_\pm$ constructed from a normed division
algebra, we need to know certain invariants of the spin group action.

\begin{lemma}
  \label{VanishingOfSymmetricInvariants}
  Let $\mathbb{K} \in \{\mathbb{R}, \mathbb{C}, \mathbb{H}, \mathbb{O}\}$ be a normed
  division algebra (Example \ref{RCHO}) of dimension $k$, and let $N_{\pm}$ be the
  real $\mathrm{Spin}(k+1,1)$ spinor representations from Proposition
  \ref{RealSpinorsFromTheNormedDivisionAlgbras}. Then
  \begin{enumerate}
    \item
      $\left(\mathrm{End}(N_{\pm})\right)^{\mathrm{Spin}}
         \simeq
         \left\{
           \begin{array}{ll}
              \mathbb{K} & \mbox{ if } \mathbb{K} \in \{\mathbb{R}, \mathbb{C}, \mathbb{H}\}
              \\
              \mathbb{R} & \mbox{ if } \mathbb{K} = \mathbb{O}
           \end{array}
         \right.
      $
    \item
       $\left(\mathrm{Sym}^2(N_{\pm})\right)^{\mathrm{Spin}} \simeq 0$
  \end{enumerate}
  where the superscript $\Spin$ denotes the subspace left invariant by
  $\Spin(k+1,1)$.
\end{lemma}

\begin{proof}
  For part 1, the algebra of $\Spin(k+1,1)$-equivariant real linear endomorphisms of
  $N_\pm$:
  \[ \End_{\Spin(k+1,1)}(N_\pm) = \left(\End(N_\pm)\right)^{\Spin} \]
  is called the \emph{commutant} of $N_\pm$. For an irreducible representation such
  as $N_\pm$, Schur's lemma tells us the commutant must be an associative division
  algebra. By the Frobenius theorem, the only associative real division algebras are
  $\R,\C$ and $\H$. We must now determine which case occurs, but this is done by
  Varadarajan \cite[Theorem 6.4.2]{Varadarajan}.

  For part 2, recall from Proposition \ref{RealSpinorsFromTheNormedDivisionAlgbras}
  that we have an invariant pairing
  \[ \langle -, - \rangle \maps N_+ \tensor N_- \to \R . \]
  Thus $N_\pm \simeq N_\mp^*$, and in particular, $\Sym^2 N_\pm \simeq \Sym^2
  N_\mp^*$. But the latter is the space of symmetric pairings:
  \[ \Sym^2 N_\mp \to \R , \]
  which is a subspace of the space of all pairings on $N_\mp$. The invariant elements
  of the space of all pairings are tabulated according to dimension and signature mod
  8 by Varadarajan \cite[Theorem 6.5.10]{Varadarajan}. In particular, for $\K = \R,
  \C$ where $N_\mp = \K^2$ are the spinors in signature $(2,1)$ and $(3,1)$
  respectively, the nonzero invariant pairings are antisymmetric, so
  $\left(\Sym^2(N_\pm)\right)^{\Spin} = 0$. For $\K = \H, \O$, where $N_\mp = \K^2$
  is the space of spinors in signature $(5,1)$ and $(9,1)$ respectively, $N_\mp$ is
  not self-dual, so there are no nonzero invariant pairings, and again we conclude
  $\left(\Sym^2(N_\pm)\right)^{\Spin} = 0$.
\end{proof}

Combining the previous two lemmas, we can prove a surprising relationship between the
Cayley--Dickson double and maximal invariant central extension: they are in essence
the same! More precisely, the maximal invariant central extension of $\R^{k+1,1|N_+
\oplus N_-}$, constructed from the normed division algebra $\K$, is given by $\R^{2k
+ 1, 1|N_{\rm dbl}}$, constructed from the Cayley--Dickson double $\K_{\rm dbl}$.

\begin{proposition}
  \label{TheExtensionsFrom3to10}
  Let $\mathbb{K} \in \{\mathbb{R}, \mathbb{C}, \mathbb{H}\}$ be an associative
  normed division algebra (Example \ref{RCHO}) of dimension $k$, and let
  $\mathbb{K}_{\mathrm{dbl}} \in \{\C, \H, \O\}$ be its Cayley--Dickson double
  (Definition \ref{CayleyDickson}) of dimension $2k$. Then the maximal invariant
  central extension of $\mathbb{R}^{k+1,1 \vert N_+ \oplus N_- }$, with $N_{\pm}$ the
  irreducible real spinor representations constructed from $\K$ as in Proposition
  \ref{RealSpinorsFromTheNormedDivisionAlgbras}, is $\mathbb{R}^{2k +1,1 \vert
  N_{\mathrm{dbl}}}$:
  $$
    \xymatrix{
      \mathbb{K} \ar@{^{(}->}[r] & \mathbb{R}^{2k+1,1 \vert N_{\mathrm{dbl}} }
      \ar[d]^-{\maximalinvariant}
      \\
      & \mathbb{R}^{k+1,1\vert N_+ \oplus N_-},
    }
  $$
  for $N_{\rm dbl}$ either of the irreducible real spinor
  representations induced by the Cayley--Dickson double $\K_{\rm dbl}$.

\end{proposition}
\begin{proof}
  By Proposition \ref{TheAutomorphismLieAlgebras}, we need to compute the even
  $\mathfrak{so}(k+1,1)$-invariant cohomology of $\mathbb{R}^{k+1,1|N_+ \oplus N_-}$
  in degree 2. Such even Lorentz invariant 2-cocycles need to pair two spinors in
  $N_+ \oplus N_-$: there is no even pairing between spinors in $N_+ \oplus N_-$ and
  vectors in $\R^{k+1,1}$, and no antisymmetric Lorentz invariant pairing between
  vectors in $\R^{k+1,1}$. Due to the simple nature of the Lie bracket on
  super-Minkowski spacetime, this means that we need to compute the space of
  $\mathfrak{so}(k+1,1)$-invariant symmetric bilinear forms on $N_+ \oplus N_-$,
  because every symmetric bilinear form on $N_+ \oplus N_-$ is an even 2-cocycle.

  We now apply Lemma \ref{branchingOfCDDoubleSpinReps} to produce these
  Lorentz-invariant 2-cocycles. Namely, let $v \in \mathbb{R}^{2k+1,1 }$ be any
  vector in the orthogonal complement of $\mathbb{R}^{k+1,1}$. Then the symmetric
  pairing
  \[
  \begin{array}{ccc}
    N_{\mathrm{dbl}} \otimes N_{\mathrm{dbl}} & \to & \mathbb{R} \\
    \psi \otimes \phi & \mapsto & \eta( v,[\psi,\phi] )
  \end{array}
  \]
  is clearly $\mathrm{Spin}(k+1,1)$-invariant, by the equivariance of the spinor
  pairing (Proposition \ref{RealSpinorsFromTheNormedDivisionAlgbras}) and the
  assumption on $v$. But by Lemma \ref{branchingOfCDDoubleSpinReps},
  $N_{\mathrm{dbl}}$ is $N_+ \oplus N_-$ as a $\mathrm{Spin}(k+1,1)$-representation.
  Therefore this construction yields a $k$-dimensional space of Spin-invariant
  symmetric bilinear pairings on $N_+ \oplus N_-$.  Moreover, by the definition of
  the pairing above, it follows that the central extension classified by these
  pairings, regarded as 2-cocycles, is $\mathbb{R}^{2k+1,1\vert
  \mathrm{N}_{\mathrm{dbl}}}$.

  To conclude the proof, it remains to show this invariant extension is maximal,
  hence that the dimension of the space of all invariant symmetric pairings on $N_+
  \oplus N_-$ is $k$. The space of all symmetric pairings, invariant or not, is:
  $$
    \mathrm{Sym}^2(N_+ \oplus N_-)
     \;\simeq\;
    \mathrm{Sym}^2(N_+)
      \;\oplus\;
    N_+ \otimes N_-
      \;\oplus\;
    \mathrm{Sym}^2(N_-)
    \,.
  $$
  So, we seek the invariant elements of the latter space.  By Lemma
  \ref{VanishingOfSymmetricInvariants}, the invariant subspaces of
  $\mathrm{Sym}^2(N_\pm)$ vanish.  Therefore the space of invariant 2-cocycles is
  the space of invariant elements in $N_+ \otimes N_-$. By the spinor-to-scalar pairing
  from Prop. \ref{RealSpinorsFromTheNormedDivisionAlgbras} the two spaces $N_+$ and
  $N_-$ are dual to each other as $\mathrm{Spin(k+1,1)}$-representations.  Therefore
  the invariant elements in $N_+ \otimes N_-$ are equivalently the
  equivariant linear endomorphisms of $N_+$:
  $$
    N_+ \to N_+
    \,.
  $$
  By Lemma \ref{VanishingOfSymmetricInvariants} this space of invariant endomorphisms is identified with $\mathbb{K}$
  $$
    \left(\mathrm{End}(N_+)\right)^{\mathrm{Spin}} \simeq \mathbb{K}
    \,.
  $$
  Hence the dimension of this space is $k$, which concludes the proof.
\end{proof}

\begin{proposition}
  \label{11dExtension}
  The maximal invariant central extension (Definition \ref{MaximalInvariantCentralExtension})
  of the type IIA super-Minkowski spacetime $\mathbb{R}^{9,1\vert \mathbf{16} \oplus \overline{\mathbf{16}}}$
  is $\mathbb{R}^{10,1\vert \mathbf{32}}$:
  $$
    \xymatrix{
      \mathbb{R} \ar@{^{(}->}[r] & \mathbb{R}^{10,1\vert \mathbf{32}}
      \ar[d]^{\maximalinvariant}
      \\
      & \mathbb{R}^{9,1\vert \mathbf{16} \oplus  \overline{\mathbf{16}}}
    }
  $$
\end{proposition}
\begin{proof}
  By Proposition \ref{TheAutomorphismLieAlgebras}, we seek even $\so(9,1)$-invariant
  2-cocycles.  Since the extension in question is clearly $\so(9,1)$-invariant, it is
  sufficient to show that the space of all $\so(9,1)$-invariant 2-cocycles on
  $\mathbb{R}^{9,1\vert \mathbf{16} + \overline{\mathbf{16}}}$ is 1-dimensional.  As
  in the proof of Proposition \ref{TheExtensionsFrom3to10}, that space is
  equivalently the space of $\so(9,1)$-invariant elements in
  \[
    \mathrm{Sym}^2(N_+ \oplus N_-)
      \simeq
    \mathrm{Sym}^2(N_+)
      \;\oplus\;
    N_+ \otimes N_-
      \;\oplus\;
    \mathrm{Sym}^2(N_-)
    \,.
  \]
  By Lemma \ref{VanishingOfSymmetricInvariants}, the invariants in
  $\mathrm{Sym}^2(N_\pm)$ vanish and the space of invariants in $N_+ \otimes N_-$ is
  one-dimensional.
\end{proof}

\noindent Putting together our results in this section, we prove our main theorem.

\begin{theorem}
  \label{thetheorem}
 The process that starts with the superpoint $\mathbb{R}^{0\vert 1}$ and then consecutively
 doubles the supersymmetries and forms the maximal invariant central extension according to
 Definition \ref{MaximalInvariantCentralExtension}
 discovers the super-Minkowski super Lie algebras $\mathbb{R}^{d-1,1\vert N}$ from Definition \ref{def:SuperMinkowski}
 in dimensions $d \in \{3,4,6,10,11\}$ for $N = 1$ and $N = 2$ supersymmetry:
 there is a diagram of super Lie algebras of the following form
$$
  \xymatrix{
    & \mathbb{R}^{10,1\vert \mathbf{32}}
      \ar[dr]|{\maximalinvariant}
    \\
     &
     \mathbb{R}^{9,1\vert \mathbf{16}}
     \ar[dr]|{\maximalinvariant}
     &
    \mathbb{R}^{9,1 \vert \mathbf{16} + \overline{\mathbf{16}}}
    \ar@{<-}@<-3pt>[l]
    \ar@{<-}@<+3pt>[l]
    \\
    &
    \mathbb{R}^{5,1\vert \mathbf{8}}
    \ar[dl]|{\maximalinvariant}
    &
    \mathbb{R}^{5,1 \vert \mathbf{8} + \overline{\mathbf{8}}}
    \ar@{<-}@<-3pt>[l]
    \ar@{<-}@<+3pt>[l]
    \\
    \mathbb{R}^{3,1\vert \mathbf{4}+ \mathbf{4}}
    \ar@{<-}@<-3pt>[r]
    \ar@{<-}@<+3pt>[r]
    &
    \mathbb{R}^{3,1\vert \mathbf{4}}
    \ar[dl]|\maximalinvariant
    \\
    \mathbb{R}^{2,1 \vert \mathbf{2} + \mathbf{2} }
    \ar@{<-}@<-3pt>[r]
    \ar@{<-}@<+3pt>[r]
    &
    \mathbb{R}^{2,1 \vert \mathbf{2}}
    \ar[dl]|{\maximalinvariant}
    \\
    \mathbb{R}^{0 \vert 1+1}
    \ar@{<-}@<-3pt>[r]
    \ar@{<-}@<+3pt>[r]
    &
    \mathbb{R}^{0\vert 1}
  }
$$
where each single arrow $\xymatrix{\ar[r]^{\maximalinvariant}& }$ denotes a maximal invariant central
extension according to Definition \ref{MaximalInvariantCentralExtension} and where
each double arrow denotes the two evident injections (Remark \ref{NumberOfSupersymmetries}).
\end{theorem}
\begin{proof}
  This is the joint statement of Proposition \ref{ProofOfTheFirstStep}, Proposition
  \ref{TheExtensionsFrom3to10} and Proposition \ref{11dExtension}.  Here we use in
  Proposition \ref{TheExtensionsFrom3to10} that for $\mathbb{K} \in \{\mathbb{R},
  \mathbb{C}\}$ a commutative division algebra, the two representations $N_{\pm}$
  from Proposition \ref{RealSpinorsFromTheNormedDivisionAlgbras} are in fact
  isomorphic.
\end{proof}

\section{Outlook}

In view of the brane bouquet \cite{FSS13}, Theorem \ref{thetheorem} is suggestive of
phenomena still to be uncovered. Further corners of M-theory, currently less well
understood, might be found by following the process of maximal invariant central
extensions in other directions. Indeed, note that Theorem \ref{thetheorem} only
exhibits some maximal invariant central extensions. It does not claim that there are
no further maximal central extensions.

For instance, the $N=1$ superpoint $\mathbb{R}^{0 \vert 1}$ also has a maximal
central extension, namely the super-line $\mathbb{R}^{1\vert 1} =
\mathbb{R}^{1,0\vert 1}$
$$
  \xymatrix{
    \mathbb{R}^{1\vert 1}
    \ar[d]^{\maximalinvariant}
    \\
    \mathbb{R}^{0 \vert 1}
  }
$$
This follows immediately with the same argument as in Proposition
\ref{ProofOfTheFirstStep}.

The natural next question is, what is the bouquet of maximal central extensions
emerging out of $\mathbb{R}^{0 \vert 3}$?  It is clear that the first step yields
$\mathbb{R}^{6\vert 3}$, with the underlying even 6-dimensional vector space
canonically identified with the $3 \times 3$ symmetric matrices with entries in the
real numbers.  Now if an analogue of Proposition \ref{TheExtensionsFrom3to10}
continued to hold in this case, then the further consecutive maximal invariant
extensions might involve the $3 \times 3$ hermitian matrices with coefficients in the
complex numbers $\mathbb{C}$, the quaternions $\mathbb{H}$, and finally the octonions
$\mathbb{O}$. The last of these, denoted $\h_3(\O)$ for the $3 \times 3$ hermitian
matrices over $\O$, is the famous exceptional Jordan algebra. Just as $\h_2(\O)$, the
$2 \times 2$ hermitian matrices over $\O$, is isomorphic to Minkowski spacetime
$\mathbb{R}^{9,1}$, so $\h_3(\O)$ is isomorphic to the 27-dimensional direct sum
$\mathbb{R}^{9,1} \oplus \mathbf{16} \oplus \mathbb{R}$ consisting of 10d-spacetime,
one copy of the real 10d spinors and a scalar \cite[Section 3.4]{Baez02}. This kind
of data is naturally associated with heterotic M-theory, and grouping its spinors
together with the vectors and the scalar to a 27-dimensional bosonic space is
reminiscent of the speculations about bosonic M-theory \cite{HorowitzSusskind00}.
Therefore, should the bouquet of maximal invariant extensions truly include
$\h_3(\O)$, this might help to better understand the nature of the bosonic or
heterotic corners of M-theory.

In a similar vein, we ought to ask how the tower of steps in Theorem \ref{thetheorem}
continues beyond dimension 11, and what the resulting structures mean.

\appendix

\section{Background}

For reference, we briefly recall some definitions and facts that we use in the main text.

\subsection{Super Lie algebra cohomology}

We recall the definition of super Lie algebras and their cohomology.  All our vector
spaces and algebras are over $\mathbb{R}$, and they are all finite dimensional. We
write \emph{even} for $0 \in \Z_2$ and \emph{odd} for $1 \in \Z_2$.

\begin{definition}
  \label{SuperLieAlg}
  The \emph{tensor category of super vector spaces} is the category of $\Z_2$-graded
  vector spaces and grade-preserving linear maps, equipped with the unique
  non-trivial braiding, $\tau^{\rm super}$. For any two super vector space $V$ and
  $W$, $\tau^{\rm super}$ is the isomorphism
  \[
  \begin{array}{rcl}
    \tau^{\rm super} \maps V \otimes W & \to & W \otimes V \\
    v_1 \otimes v_2 & \mapsto & (-1)^{\sigma_1 \sigma_2} \, v_2 \otimes v_1 ,
     \end{array}
\]
for elements $v_1 \in V$, $v_2 \in W$ of homogeneous degree $\sigma_i \in \Z_2$.

  A \emph{super Lie algebra} is a Lie algebra internal to super vector spaces. That
  is, it is a super vector space
  $$
    \mathfrak{g} = \mathfrak{g}_{\mathrm{even}} \oplus \mathfrak{g}_{\mathrm{odd}}
  $$
  equipped with a bilinear map, called the \emph{Lie bracket}:
  $$
    [-,-] \maps \mathfrak{g} \otimes \mathfrak{g} \longrightarrow \mathfrak{g}
  $$
  which is graded skew symmetric:
  $$
    [v_1, v_2] = -(-1)^{\sigma_1 \sigma_2} [v_2, v_1]
  $$
  and which satisfies the graded Jacobi identity:
  $$
    [v_1, [v_2,v_3]] = [[v_1, v_2], v_3] + (-1)^{ \sigma_1 \sigma_2 } [v_2, [v_1, v_3]]
    \,.
  $$
  A homomorphism of super Lie algebras $\mathfrak{g}_1 \longrightarrow
  \mathfrak{g}_2$ is a linear map preserving the $\Z_2$-grading and the bracket.
\end{definition}

\begin{definition}[super Lie algebra cohomology]
  \label{CEAlgebra}
  Let $V$ be a finite-dimensional super vector space. Then the \emph{super-Grassmann
  algebra} $\Lambda^\bullet V^\ast$ is the $\mathbb{Z} \times
  (\Z_2)$-bigraded-commutative associative algebra freely generated by $V^*$ in
  degree $1 \in \Z$. That is to say it is generated by the elements in
  $V^\ast_{\mathrm{even}}$ regarded as being in bidegree $(1,\mathrm{even})$, and the
  elements in $V^\ast_{\mathrm{odd}}$ regarded as being in bidegree
  $(1,\mathrm{odd})$, subject to the relation that for $\alpha_i$ two elements of
  homogeneous bidegree $(n_i, \sigma_i)$, then
  $$
    \alpha_1 \wedge \alpha_2
      \;=\;
    (-1)^{n_1 n_2} (-1)^{\sigma_1 \sigma_2} \, \alpha_2 \wedge \alpha_1
    \,.
  $$
  In particular, this relation says that elements of bidegree $(1,{\rm even})$
  anticommute with each other, those of bidegree $(1,{\rm odd})$ commute with each
  other, while an element of bidegree $(1,{\rm even})$ anticommutes with an element
  of bidegree $(1,{\rm odd})$.

  Now let $(\mathfrak{g}, [-,-])$ be a finite-dimensional super Lie algebra. Then its
  \emph{Chevalley--Eilenberg algebra} $\mathrm{CE}(\mathfrak{g})$ is the
  super-Grassmann algebra $\Lambda^\bullet \mathfrak{g}^\ast$ equipped with the
  differential $d_{\rm CE}$ defined as follows. On the generators
  $\mathfrak{g}^\ast$, $d_{\rm CE}$ acts as the linear dual of the Lie bracket:
  $$
    [-,-]^* \maps \mathfrak{g}^\ast \to \Lambda^2 \mathfrak{g}^\ast.
  $$
  The action of $d_{\rm CE}$ on generators is then extended to all of
  $\Lambda^\bullet \g^*$ as a derivation, bigraded of bidegree $(1,\even)$. This makes
  $\CE(\g)$ into a differential graded algebra. A calculation shows $d_{\rm CE}^2 =
  0$, so $\CE(\g)$ is also a cochain complex. 

  For $p \in \mathbb{N}$ we say that a \emph{$(p+2)$-cocycle} on $\mathfrak{g}$ with
  coefficients in $\mathbb{R}$ is a $d_{\mathrm{CE}}$-closed element in
  $\Lambda^{p+2} \mathfrak{g}^\ast$. We say that cocycle is \emph{even} if its degree
  in $\Z_2$ is even, and \emph{odd} if it is odd. The \emph{super Lie algebra
  cohomology} of $\mathfrak{g}$ with coefficients in $\mathbb{R}$ is the cohomology
  of its Chevalley--Eilenberg algebra, regarded as a cochain complex:
  $$
    H^\bullet(\mathfrak{g}, \mathbb{R})
    :=
    H^\bullet(\mathrm{CE}(\mathfrak{g}))
    \,.
  $$
  The $\Z_2$-grading on $\CE(\g)$ makes $H^p(\g, \R)$ into a super vector space for each
  $p$. We will be interested in its even part, $H^p_{\even}(\g, \R)$.
\end{definition}

\begin{example}
  \label{2CocyclesAndCentralExts}
  Let $\mathfrak{g}$ be a finite dimensional super Lie algebra, and let $\omega \in
  \Lambda^2 \mathfrak{g}^\ast$ be an even 2-cocycle as in Definition
  \ref{CEAlgebra}. Then there is a new super Lie algebra $\widehat{\mathfrak{g}}$
  whose underlying super vector space is
  $$
    \widehat{\mathfrak{g}}
      :=
    \underset{\mathrm{even}}{\underbrace{\mathfrak{g}_{\mathrm{even}} \oplus \mathbb{R}}}
     \oplus
    \underset{\mathrm{odd}}{\underbrace{\mathfrak{g}_{\mathrm{odd}}}}
  $$
  and with super Lie bracket given by
  $$
    [ (x_1,c_1), (x_2,c_2) ]
    \;=\;
    ([x_1,x_2], \omega(x_1, x_2))
    \,.
  $$
  We thus have a short exact sequence giving $\widehat{\g}$ as a central extension of
  $\g$:
  \[ 0 \to \R \to \widehat{\g} \to \g \to 0 . \]
  In the paper, we will often write this short exact sequence as follows, in the
  style an algebraic topologist might use to write down a fibration:
  $$
    \xymatrix{
      \mathbb{R} \ar@{^{(}->}[r] & \widehat{\mathfrak{g}}
      \ar[d]
      \\
      & \mathfrak{g} \, .
    }
  $$
  Just as for ordinary Lie algebras, this construction establishes a natural
  equivalence between central extensions of $\mathfrak{g}$ by $\mathbb{R}$ (in even
  degree) and even super Lie algebra 2-cocycles on $\mathfrak{g}$.

  More generally, a central extension in even degree is by some vector space $V
  \simeq \mathbb{R}^n$
  $$
    \xymatrix{
      \mathbb{R}^n \ar@{^{(}->}[r] & \widehat{\mathfrak{g}}
      \ar[d]
      \\
      & \mathfrak{g}
    }
  $$
  which is equivalently the result of extending by $n$ even 2-cocycles, one after the
  other, in any order.
\end{example}

We will be interested not in the full super Lie algebra cohomology, but in the \emph{invariant} cohomology
with respect to some action:

\begin{definition}
  \label{AutomorphismLieGroup}
  For $\mathfrak{g}$ a super Lie algebra (Definition \ref{SuperLieAlg}), its
  \emph{automorphism group} is the Lie subgroup
  $$
    \mathrm{Aut}(\mathfrak{g})
      \hookrightarrow
    \mathrm{GL}(\mathfrak{g}_{\mathrm{even}})
    \times
    \mathrm{GL}(\mathfrak{g}_{\mathrm{odd}})
  $$
  of the group of degree-preserving linear isomorphisms on the underlying vector
  space, consisting of those elements which are super Lie algebra isomorphisms. 
\end{definition}
\begin{proposition}
  \label{automorphismLieAlg}
  For $\mathfrak{g}$ a super Lie algebra, the Lie algebra of its automorphism Lie
  group (Definition \ref{AutomorphismLieGroup})
  $$
    \mathfrak{aut}(\mathfrak{g})
  $$
  is called the the \emph{automorphism Lie algebra} of $\mathfrak{g}$. It is the Lie
  algebra of those linear maps $\Delta \maps \mathfrak{g} \to \mathfrak{g}$ which
  preserve the degree and satisfy the derivation property:
  $$
    \Delta[X,Y] = [\Delta X, Y] + [X, \Delta Y]
  $$
  for all $X, Y\in \mathfrak{g}$. The Lie bracket on $\mathfrak{aut}(\mathfrak{g})$ is
  the commutator:
  $$
    [\Delta_1, \Delta_2] := \Delta_1 \Delta_2 - \Delta_2 \Delta_1
    \,.
  $$
\end{proposition}
We caution the reader that, even though $\g$ is a super Lie algebra, its automorphism
algebra $\aut(\g)$ is merely a Lie algebra. This is because we want elements of
$\aut(\g)$ to preserve the degree on $\g$. 
\begin{example}
  \label{scaling}
  The super-Minkowski super Lie algebras $\mathbb{R}^{d-1,1\vert N}$ from Definition \ref{def:SuperMinkowski}
  all carry an automorphism action of the abelian Lie algebra $\mathbb{R}$ which is spanned
  by the \emph{scaling derivation} that acts on vectors $v \in \mathbb{R}^{d-1,1}$ by
  $$
    v \mapsto 2v
  $$
  and on spinors $\psi \in N$ by
  $$
    \psi \mapsto \psi
    \,.
  $$
\end{example}
\begin{definition}
  \label{InvariantCohomology}
  Let $\g$ be a super Lie algebra (Def. \ref{SuperLieAlg}). Clearly, every
  automorphism of $\g$ will induce an automorphism of the Chevalley--Eilenberg
  algebra $\CE(\g)$ (Def. \ref{CEAlgebra}). Explicitly, this works as follows. Let
  $\Delta \in \aut(\g)$ be an infinitesimal automorphism (Prop.
  \ref{automorphismLieAlg}). The induced automorphism $\Delta_\CE \maps \CE(\g) \to
  \CE(\g)$ acts on the generators $\g^*$ of $\CE(\g)$ as the linear dual $\Delta^*$:
  $$
    \Delta_{\mathrm{CE}} \maps \mathfrak{g}^\ast \overset{\Delta^\ast}{\longrightarrow} \mathfrak{g}^\ast.
  $$
  This is then extended to all of $\CE(\g)$ as a derivation of bidegree
  $(0,\even)$. The fact that $\Delta_\CE$ commutes with $d_\CE$ is equivalent to the
  fact that $\Delta$ is a derivation of $\g$.

  Now let $\mathfrak{h} \hookrightarrow \mathfrak{aut}(\mathfrak{g})$ be a Lie
  subalgebra of its automorphism Lie algebra. The elements of
  $\mathrm{CE}(\mathfrak{g})$ which are annihilated by $\Delta_{\mathrm{CE}}$ for all
  $\Delta \in \mathfrak{h}$ form a differential graded subalgebra of $\CE(\g)$:
  $$
    \mathrm{CE}(\mathfrak{g})^{\mathfrak{h}}
      \hookrightarrow
    \mathrm{CE}(\mathfrak{g})
  $$

  We say an \emph{$\mathfrak{h}$-invariant $(p+2)$-cocycle} on $\mathfrak{g}$ is an
  element in $\mathrm{CE}(\mathfrak{g})^{\mathfrak{h}}$ which is
  $d_{\mathrm{CE}}$-closed and the \emph{$\mathfrak{h}$-invariant cohomology} of
  $\mathfrak{g}$ with coefficients in $\mathbb{R}$ is the cochain cohomology of this
  subcomplex:
  $$
    H^\bullet(\mathfrak{g}, \mathbb{R})^{\mathfrak{h}}
     :=
    H^\bullet(\mathrm{CE}(\mathfrak{g})^{\mathfrak{h}})
    \,.
  $$
  We define even and odd invariant cocycles as before. The vector space $H^p(\g,
  \R)^\h$ is $\Z_2$-graded for each $p$, and our focus will be on its even part,
  $H^p_\even(\g,\R)^\h$.
\end{definition}

\subsection{Super-Minkowski spacetimes}
\label{sec:SuperMinkowski}

We recall the definition of `super-Minkowski super Lie algebras' (Definition
\ref{def:SuperMinkowski}) as well as their construction, on the one hand via
Majorana or symplectic Majorana spinors (Example \ref{MajoranaRepresentations}), and on the
other hand via the four normed division algebras (Proposition
\ref{RealSpinorsFromTheNormedDivisionAlgbras}). We freely use basic facts about
spinors, as may be found in the book of Lawson and Michelsohn \cite{LawsonMichelsohn89}.

\begin{definition}[super-Minkowski Lie algebras]
  \label{def:SuperMinkowski}
  Let $d \in \mathbb{N}$ (spacetime dimension) and let $N$ be a real spinor
  representation of $\Spin(d-1,1)$, the double cover of the connected Lorentz group
  $\SO_0(d-1,1)$.  Then \emph{$d$-dimensional $N$-supersymmetric super-Minkowski
  spacetime} $\mathbb{R}^{d-1,1\vert N}$ is the super Lie algebra (Definition
  \ref{SuperLieAlg}) whose underlying super-vector space is
  $$
    \mathbb{R}^{d-1,1\vert N}
      :=
    \underset{\mathrm{even}}{\underbrace{\mathbb{R}^{d-1,1}}} \oplus \underset{\mathrm{odd}}{\underbrace{N}}.
  $$
  The Lie bracket is nonzero only on $N$, and is a choice of symmetric, bilinear,
  $\Spin(d-1,1)$-equivariant map:
  $$
    [-,-] \maps N \otimes N \longrightarrow \mathbb{R}^{d-1,1} .
  $$
  Such a map is always available in spacetime signature $(d-1,1)$
  \cite{FreedLectures}, though there may be more than one choice \cite{Varadarajan}.

  There is a canonical action of $\mathrm{Spin}(d-1,1)$ on $\mathbb{R}^{d-1,1\vert N}$
  by Lie algebra automorphisms, and the corresponding semidirect product Lie algebra
  is the \emph{super Poincar{\'e} super Lie algebra}
  $$
    \mathfrak{iso}(\mathbb{R}^{d-1,1\vert N})
    =
    \mathbb{R}^{d-1,1\vert N}
      \rtimes
    \mathfrak{so}(d-1,1)
    \,.
  $$
  It is also called the \emph{supersymmetry algebra}. 
\end{definition}
\begin{remark}[number of super-symmetries]
  \label{NumberOfSupersymmetries}
  In the physics literature the choice of real spinor representation in Definition \ref{def:SuperMinkowski}
  is often referred to as the `number of supersymmetries'. While this is imprecise,
  it fits well with the convention of labelling irreducible representations by their
  dimension in boldface. For example when $d = 10$ there are two irreduible real
  spinor representations, both of real dimension 16, but of opposite chirality, and hence
  traditionally denoted $\mathbf{16}$ and $\overline{\mathbf{16}}$. Hence we may speak of
  $N = \mathbf{16}$ supersymmetry (also called $N = 1$, \emph{type I} or \emph{heterotic}) and
  $N = \mathbf{16} \oplus \overline{\mathbf{16}}$ supersymmetry (also called $N = (1,1)$ or \emph{type IIA})
  and $N = \mathbf{16} \oplus \mathbf{16}$ supersymmetry (also called $N = (2,0)$ or \emph{type IIB}).

  In Section \ref{TheConsecutiveExteension} the generalization of the last of these cases plays a central role, where for any
  given real spin representation $N$ we pass to the doubled supersymmetry $N \oplus N$. Observe that the
  two canonical linear injections $N \to N \oplus N$ into the direct sum induce two super Lie algebra
  homomorphisms
  $$
    \xymatrix{
      \mathbb{R}^{d-1,1\vert N}
      \ar@<-3pt>[r]
      \ar@<+3pt>[r]
      &
      \mathbb{R}^{d-1,1\vert N \oplus N}
    }
    \,.
  $$
\end{remark}
The following degenerate variation of super-Minkowski spacetime will play a key role:
\begin{definition}[superpoint]
  \label{Superpoint}
  A \emph{superpoint} is the super Lie algebra
  $$
    \mathbb{R}^{0|N}
  $$
  which has zero Lie bracket, and whose underlying super vector space is concentrated
  in odd degree, where it is of dimension $N$. 
\end{definition}

We will use two different ways of constructing real spin representations, and hence
super-Minkowski spacetimes: via `Majorana' or `symplectic Majorana' spinors (Example
\ref{MajoranaRepresentations}) and via real normed division algebras (Proposition
\ref{RealSpinorsFromTheNormedDivisionAlgbras}).

\begin{example}[Majorana representations]
  \label{MajoranaRepresentations}
  For $d = 2\nu$ or $2 \nu+1$, there exists a complex representation of the Clifford algebra
  $\mathrm{Cl}(\mathbb{R}^{d-1,1})\otimes \mathbb{C}$, hence of the spin group $\mathrm{Spin}(d-1,1)$
  on $\mathbb{C}^{2^\nu}$
  such that
  \begin{enumerate}
    \item all skew-symmetrized products of $p \geq 1$ Clifford elements $\Gamma_{a_1 \cdots a_p}$ are traceless;
    \item $\Gamma_0^\dagger = \Gamma_0$ and $\Gamma_i^\dagger = - \Gamma_i$, for $1 \leq i \leq d - 1$.
  \end{enumerate}
  This is the \emph{Dirac representation}, a complex representation of $\Spin(d-1,1)$.
  For $d = 2\nu$ this is the direct sum of two subrepresentations on $\mathbb{C}^{2^{\nu}-1}$,
  the \emph{Weyl representations}.

  For $d \in \{1,2,3,4,8,9,10,11\}$, there exists a real structure $J$ on the
  complex Dirac representation, restricting to the Weyl representations for $d = 2$
  or $d = 10$. This is a $\mathrm{Spin}(d-1,1)$-equivariant antilinear endomorphism
  $J \maps S \to S$ which squares to the identity: $J^2 = +1$.  It carves out a real
  representation called the \emph{Majorana representation} $N := \mathrm{Eig}(J,+1)$,
  the eigenspace of $J$ of eigenvalue +1, whose elements are called the
  \emph{Majorana spinors}. In this case the \emph{Dirac conjugation} $\psi \mapsto
  \psi^\dagger \Gamma_0$ on elements $\psi \in \mathbb{C}^{2^\nu}$ restricts to $N$
  and is called the \emph{Majorana conjugation}. We write it as simply $\psi \mapsto
  \overline{\psi}$.  In terms of this matrix representation then the spinor bilinear
  pairing that appears in Definition \ref{def:SuperMinkowski} is given by the
  following matrix product expression:
$$
  [\psi,\phi] \;=\; \left(\overline{\psi} \Gamma^a \phi\right)_{a = 0}^{d-1}
  \,.
$$

Similarly, for $d \in \{5,6,7\}$ there exists a quaternionic structure on the Dirac
representation. This is a $\Spin(d-1,1)$-equivariant antilinear endomorphism $\tilde
J$ which squares to minus the identity, $\tilde J^2 = -1$.  It follows that
$$
  J := \left(
    \begin{array}{cc}
      0 & -\tilde J
      \\
      \tilde J & 0
    \end{array}
  \right)
$$
is a real structure on the direct sum of the Dirac representation with itself. Hence
as before $N := \mathrm{Eig}(J,+1)$ is a real subrepresentation, called the
\emph{symplectic Majorana representation}. The spinor-to-vector bilinear pairing for
symplectic Majorana spinors is similar to the case of Majorana spinors.
\end{example}

\begin{definition}[Cayley--Dickson double {\cite[Section 2.2]{Baez02}}]
  \label{CayleyDickson}
  Let $\mathbb{K}$ be a real $*$-algebra. This is a real, not necessarily
  associative algebra $\K$ equipped with a \emph{conjugation}
  $\overline{(-)} \maps \mathbb{K} \to \mathbb{K}$, satisfying:
  \[ \overline{a + b} = \overline{a} + \overline{b}, \quad \overline{ab} =
  \overline{b} \, \overline{a}, \quad \overline{\overline{a}} = a , \]
  for any $a, b \in \K$. 
  Then the \emph{Cayley--Dickson double} $\mathbb{K}_{\rm dbl}$ of $\mathbb{K}$ is the real
  $*$-algebra obtained from $\mathbb{K}$ by adjoining one element $\ell$ such that
  $\ell^2 = -1$ and such that the following relations hold, for all $a, b \in
  \mathbb{K}$:
  $$
    a(\ell b) = \ell(\overline{a} b)
    \,,\;\;\;\;
    (a \ell) b = (a \overline{b}) \ell
    \,,\;\;\;\;
    (\ell a)(b \ell) = -\, \overline{(a b)}
    \,.
  $$
  Finally, the conjugation $\overline{(-)}$ on $\mathbb{K}_{\mathrm{dbl}}$ acts on
  elements of $\K$ by the conjugation on $\K$, and sends the new generator $\ell$ to
  $-\ell$.
\end{definition}
\begin{example}
  \label{RCHO}
  Consider $\mathbb{R}$ the real numbers regarded as a $*$-algebra with trivial
  conjugation $\overline{a} = a$. Then its Cayley--Dickson double (Definition
  \ref{CayleyDickson}) is the complex numbers $\mathbb{C}$ with the usual
  conjugation, the Cayley--Dickson double of $\C$ is the quaternions $\mathbb{H}$,
  and the Cayley--Dickson double of $\H$ is the octonions $\mathbb{O}$.

  By a classical result of Hurwitz, these four algebras are the only normed division
  algebras over the real numbers, as reviewed by Baez \cite{Baez02}.
\end{example}

In the next proposition and elsewhere in the text, we will use $n \times n$ matrices
over $\K$ to describe real linear operators on $\K^n$. We will write $\K[n]$ for the
set of all $n \times n$ matrices with entries in $\K$. For any such matrix, there are
two natural ways for it to induce a linear operator, one using left multiplication in
$\K$ and the other right multiplication.

\begin{definition}(Matrices over $\K$ as linear operators)
  \label{MatricesToLinearOps}
Let $\K$ be a normed division algebra. Any element of $a \in \K$ induces a linear
endomorphism on $\K$ by left or right multiplication, which we will denote by $a_L$
or $a_R$, respectively:
\[ \begin{array}{rcl}
  a_L \maps \K & \to & \K \\
  x & \mapsto & ax
  \end{array}
  , \quad
  \begin{array}{rcl}
  a_R \maps \K & \to & \K \\
  x & \mapsto & xa .
  \end{array}
\]
More generally, any $n \times n$ matrix $A \in \K[n]$ induces a linear endomorphism on
$\K^n$ via either left multiplication or right multiplication:
\[ \begin{array}{rcl}
  A_L \maps \K^n & \to & \K^n \\
  x & \mapsto & \sum a_{ij} x_j
  \end{array}
, \quad
\begin{array}{rcl}
  A_R \maps \K^n & \to & \K^n \\
  x & \mapsto & \sum x_j a_{ij}
  \end{array}
\]
where we are using the subscript $x_j$ to denote the $j$th coordinate of $x \in
\K^n$, and $A = (a_{ij})$. In other words, $A_L$ and $A_R$ are the linear maps on
$\K^n$ with components $\left( (a_{ij})_L \right)$ and $\left( (a_{ij})_R \right)$,
respectively. We say that $A_L$ is the \emph{left action} of $A$ and $A_R$ is the
\emph{right action} of $A$. We caution that because $\K$ is nonassociative, $(AB)_L
\neq A_L B_L$ in general, and because $\K$ is nonassociative and noncommutative,
$(AB)_R \neq A_R B_R$ in general.
\end{definition}

\begin{remark}
  The left action of a matrix $A$ by $A_L$ is just the usual matrix multiplication,
  so we will sometimes write:
  \[ A_L x = Ax . \]
  The utility of defining the linear transformation $A_L$ is that the composition of
  linear transformations is associative, so we do not need to worry about the
  nonassociativity of $\K$ when we compose them. For example:
  \[ A_L B_L C_L x = A(B(Cx)) . \]
\end{remark}

Since $\K$ comes with a conjugation, we can define the conjugate of any matrix in
$\K[n]$ by taking the conjugate of each entry, and the conjugate of any element of
$\K^n$ by taking the conjugate of each coordinate. It is then an elementary
calculation to show that the action of a matrix $A$ by $A_L$ and $A_R$ are related by
conjugation:

\begin{proposition}
  \label{LeftAndRight}
  Let $A \in \K[n]$ be an $n \times n$ matrix over the normed division algebra $\K$
  (as defined in Example \ref{RCHO}). Then
  \[ \overline{A_L x} = (\overline{A})_R \, \overline{x} \mbox{ and } \overline{A_R x} = (\overline{A})_L \, \overline{x} \]
  for all $x \in \K^n$.
\end{proposition}

The next definition is straightforward, but is central to realizing spin
representations via normed division algebras.
\begin{definition}[{\cite{Schray96}}]
  \label{TraceReversal}
  For $A \in \h_2(\K)$ a hermitian matrix with coefficients in
  one of the four real normed division algebras from Example \ref{RCHO}. Then its \emph{trace reversal}
  is
  $$
    \widetilde A
      :=
    A - \mathrm{tr}(A)\cdot \mathbf{1}
    \,.
  $$
\end{definition}
\begin{proposition}[{\cite{BaezHuerta09}}]
  \label{RealSpinorsFromTheNormedDivisionAlgbras}
  Let $\K \in \{\mathbb{R}, \mathbb{C}, \mathbb{H}, \mathbb{O}\}$ be one of the
  normed division algebras as in Example \ref{RCHO}.  Write $\h_2(\K)$ for the real
  vector space of $2 \times 2$ hermitian matrices with coefficients in $\K$, and $k$
  for the dimension of $\K$.

  Then:
  \begin{enumerate}
    \item There is an isomorphism of inner product spaces (``forming Pauli matrices over $\mathbb{K}$'')
    $$
      (\mathbb{R}^{k+1,1}, \eta)
        \stackrel{\simeq}{\longrightarrow}
      \left(
        \h_2(\K), -\mathrm{det}
      \right)
    $$
    identifying $\mathbb{R}^{k+1,1}$ equipped with its Minkowski inner product
    $$
      \eta(A,B) := -A^0 B^0 + A^1 B^1 + \cdots + A^{k+1} B^{k+1}, \mbox{ for } A, B \in \R^{k+1,1} 
    $$
    with the space of hermitian matrices equipped with the negative of the determinant operation.

  \item Let $N_+$ and $N_-$ both denote the vector space $\K^2$. Then $N_+ \oplus
    N_-$ is a module of the Clifford algebra $\Cl(k+1,1)$, with the action of a
    vector in $A \in \R^{k+1,1}$ given by
    \[ \Gamma(A) (\psi, \phi) = (\tilde{A}_L \phi, A_L \psi) \]
    for any element $(\psi, \phi) \in N_+ \oplus N_-$, where we are using the
    identification of vectors with $2 \times 2$ hermitian matrices. Here $\widetilde
    {(-)}$ is the trace reversal operation from Def. \ref{TraceReversal}, and
    $(-)_L$ denotes the linear map given by left multiplication as in
    Def. \ref{MatricesToLinearOps}

  \item Realizing the spin group $\Spin(k+1,1)$ inside the Clifford algebra
    $\Cl(k+1,1)$ by the standard construction, this induces irreducible
    representations $\rho_\pm$ of $\Spin(k+1,1)$ on $N_\pm$. Explicitly, recall that
    $\Spin(k+1,1)$ is the subgroup of the Clifford algebra
    generated by products of pairs of unit vectors of the same sign:
      \[ \Spin(k+1, 1) = \langle AB \in \Cl(k+1,1) \, : \, A, B \in \R^{k+1,1}, \, \eta(A,A) = \eta(B,B) = \pm 1
  \rangle . \]
    Then restricting the Clifford action to these elements, a generator $AB$ of
    $\Spin(k+1,1)$ acts as
    \[ \rho_+(AB) = \tilde{A}_L B_L \mbox{ on } N_+ \]
    and as
    \[ \rho_-(AB) = A_L \tilde{B}_L \mbox{ on } N_-,\]
    where again $\widetilde {(-)}$ is the trace reversal operation from
    Def. \ref{TraceReversal}, and where $(-)_L$ denotes the linear map given by left
    multiplication as in Def. \ref{MatricesToLinearOps}.

    For $\mathbb{K} \in \{\mathbb{R}, \mathbb{C}\}$ then these two representations
    are in fact isomorphic and are the \emph{Majorana representation} of
    $\mathrm{Spin}(2,1)$ and $\Spin(3,1)$, respectively, while for $\mathbb{K} \in \{
    \mathbb{H}, \mathbb{O} \}$ they are the two non-isomorphic
    \emph{symplectic-Majorana representations} of $\mathrm{Spin}(5,1)$ and
    \emph{Majorana--Weyl representations} of $\mathrm{Spin}(9,1)$, respectively.

  \item Under the above identifications, the symmetric bilinear $\Spin(k+1,1)$-equivariant
    spinor-to-vector pairings are given by
    \[
    \begin{array}{rcl}
      [-,-] \maps N_+ \otimes N_+ & \to & \R^{k+1,1} \\
      \psi \otimes \phi & \mapsto & \tfrac{1}{2}\left( \widetilde{ \psi \phi^\dagger + \phi \psi^\dagger} \right)
      \,
      \end{array}
    \]
    and
    \[
    \begin{array}{rcl}
      [-,-] \maps N_- \otimes N_- & \to & \R^{k+1,1} \\
      \psi \otimes \phi & \mapsto & \tfrac{1}{2}\left( \psi \phi^\dagger + \phi \psi^\dagger \right)
      \,
      \end{array}
    \]

  \item There is a bilinear symmetric, non-degenerate and $\mathrm{Spin}(k+1,1)$-invariant
    spinor-to-scalar pairing given by
    \[
    \begin{array}{rcl}
      \langle -, -\rangle \maps N_\pm \otimes N_{\mp}  & \to & \R \\
      \psi \otimes \phi & \mapsto & \mathrm{Re}(\psi^\dagger \phi)
      \,.
    \end{array}
    \]
  \end{enumerate}
\end{proposition}

\medskip

\noindent{\bf Acknowledgements.}
We thank David Corfield for comments on an earlier version of this article. We also
thank John Baez and an anonymous referee for comments and corrections. We thank the
Max Planck Institute for Mathematics in Bonn for kind hospitality while the result
reported here was conceived. U.S.\ thanks Roger Picken for an invitation to Instituto
Superior T\'ecnico, Lisbon, where parts of this article were written. U.S.\ was
supported by RVO:67985840. J.H.\ was supported by the Portuguese science foundation
grant SFRH/BPD/92915/2013.

\end{document}